\documentclass[twoside,a4paper,12pt,reqno]{amsart}
\usepackage{latexsym}
\usepackage{amsmath}
\usepackage{amssymb}
\usepackage{amsthm}
\usepackage{amsxtra}
\usepackage{amsfonts}
\usepackage{mathrsfs}
\usepackage{amsbsy}

\newtheorem{definizione}{Definition}[section]
\newtheorem{teorema}{Theorem}[section]
\newtheorem{prop}{Proposition}[section]
\newtheorem{lemma}{Lemma}[section]

\newcommand{\n}{\noindent}

\newcommand{\ve}{\varepsilon}
\newcommand{\erre}{\mathbb{R}}

\newcommand{\enne}{\mathbb{N} }

\newcommand{\f}{\frac}
\newcommand{\ba}{\begin{eqnarray}} \newcommand{\ea}{\end{eqnarray}}
\newcommand{\be}{\begin{equation}} \newcommand{\ee}{\end{equation}}
\newcommand{\bdm}{\begin{displaymath}} \newcommand{\edm}{\end{displaymath}}
\newcommand{\brr}{\begin{array}}\newcommand{\err}{\end{array}}

\newcommand{\bml}{\begin{gather}}
\newcommand{\eml}{\end{gather}}

\newcommand{\vs}{\vspace{.5cm}}
\newcommand{\bo}{\boldsymbol}

\newcommand{\bs}{\boldsymbol}

\newcommand{\mue}{\mathcal U^\varepsilon }
\newcommand{\mul}{\mathcal U_0^\varepsilon }

\newcommand{\un}{\underline{n}}
\newcommand{\um}{\underline{m}}
\newcommand{\uo}{\underline{0}}

\numberwithin{equation}{section}

\begin{document}


\begin{center}
{\Huge\bf Semiclassical wave-packets emerging \\
\vspace{0.2cm} from interaction with an environment  }
 
\vs
\vs
\vs

{\large\bf Carla Recchia$^1$ and Alessandro Teta$^2$}

\vs
\vs

{\small $^1$D.I.S.I.M., Universit\`a di L'Aquila} \\
\small{Via Vetoio - Loc. Coppito - 67010 L'Aquila, Italy, carla.recchia@libero.it}
\\
\small{  $^2$Dipartimento di Matematica, ''Sapienza'' Universit\`a di Roma}\\
\small{P.le A. Moro 5, 00185 Roma, Italy,  teta@mat.uniroma1.it  
} 

\end{center}

\vs
\vs\vs
\vs

\begin{center}
{\bf Abstract} 
\end{center}

\n
We study the   quantum evolution  in dimension three of a system composed by a test particle interacting with an environment made of $N$ harmonic oscillators. At time zero the test particle is described by a spherical wave,  i.e. a highly correlated continuous superposition of states with well localized position and momentum,  and the oscillators are in the ground state. Under suitable assumptions on the physical parameters characterizing the model, we give an asymptotic expression of the solution of the Schr\"odinger equation of the system with an explicit control of the error. The result shows that the approximate expression of the wave function is the sum of   two  terms, orthogonal in $L^2(\erre^{3(N+1)})$ and describing  rather different situations.

\n
In the first one all the oscillators remain in their ground state and the test particle is described by the free evolution of a slightly deformed spherical wave. 

\n
The second one consists of a sum of $N$ terms where in each term  there is only one excited  oscillator and the test particle  is correspondingly described by the free evolution of a  wave packet, well concentrated in position and momentum. Moreover the wave packet emerges from the excited oscillator with an average momentum parallel to the line joining  the oscillator with the center of the initial spherical wave.  
Such wave packet represents a semiclassical state for the test particle, propagating along the corresponding  classical trajectory.

\n
The main result of our analysis  is to show how such a semiclassical state can be produced, starting from the original  spherical wave,  as a result of the interaction with the environment. 

\newpage

\section{Introduction}

\n
The analysis of the  emergence of a classical behavior in a quantum system is a subject of interest not only from the conceptual but also from the applicative point of view. Indeed,  for the correct behavior of any quantum device it is surely relevant to   understand  under which conditions genuine quantum effects, like interference or entanglement, are reduced or cancelled.

\n
It is widely accepted that a crucial classicality condition  is obtained requiring that  the Planck's constant $\hbar$ is small with respect to the typical action of the system. 
An important example where the classical behavior is recovered in the limit  $\hbar \rightarrow 0$ is the case  of  the evolution of a quantum particle initially described by a ''semiclassical'' state, e.g. a state with a suitable dependence on $\hbar$ like a WKB or a coherent state. 
A typical result in this context  is that for $\hbar \rightarrow 0$ the evolution is again given by a semiclassical state, propagating along the corresponding classical trajectories (see e.g. \cite{r}, \cite{cr}).

\n
On the other hand, it is important to remark that in many interesting physical  situations the initial state of the quantum system is a genuine quantum state, like a superposition state, and nevertheless the system exhibits  a classical behavior during the evolution. In such cases  the limit $\hbar \rightarrow 0$ for the isolated system cannot help and the role of the quantum environment  must be taken into account, according to the approach known as decoherence theory (see e.g. \cite{gjkksz}, \cite{bgjks}, \cite{h}; for some rigorous results see \cite{afft}, \cite{ccf}). Roughly speaking, for such kind of physical situations one has to consider the Schr\"odinger equation for the ''system + environment'' and one has to prove that, for appropriate values of the physical parameters characterizing system and environment, the system shows a classical behavior as a result of its interaction with the environment. 

\n
In this paper we consider the case of the evolution of a quantum particle initially described by a spherical wave, i.e. a highly correlated continuous superposition of states with well localized position and momentum. Our aim is to show that a semiclassical wave packet for the particle, and therefore the propagation along a classical path, emerges as a result of the interaction with the environment.

\n
This kind of problem was already raised in the early days of Quantum Mechanics and it was discussed in a seminal paper by Mott in 1929 (\cite{m}) in connection with a possible explanation of the linear tracks left by an $\alpha$-particle in a cloud chamber (see also \cite{figt} for a historical analysis of the problem). We recall that the $\alpha$-particle is emitted by a radioactive source in the form of a spherical wave and then it interacts with the atoms of the vapor filling the device.  If an atom is excited then the amplified effect of the excitation is observed. As a matter of fact, in the experiment one observes straight tracks, corresponding to the excitation of a sequence of many atoms whose positions are aligned with the center of the spherical wave. The observed track is regarded as the experimental manifestation of the ''classical trajectory''  of the $\alpha$-particle. From the point of view of a theoretical description, the non trivial problem arises to explain how a spherical wave can produce the observed classical trajectory.  In Mott's paper an answer is given considering  a model of environment made of only two atoms. At a physical level of rigor, it is shown that the probability that both atoms are excited is negligible unless the two atoms are aligned with the center of the spherical wave. The approach is based on  second order perturbation theory for the stationary Schr\"odinger equation and a repeated use of stationary phase arguments (we refer to \cite{dft} for an attempt to 
revisit Mott's approach in a fully time dependent setting and to \cite{fint} for a detailed analysis of a simpler one dimensional model).  

\n
Here we consider a more general quantum system  in $\erre^3$ made of a test particle, initially described by a spherical wave centered in the origin, interacting with  $N$ harmonic oscillators, initially placed in their ground state.  The aim is to study the evolution of the whole system after the interaction of the test particle with each oscillator has taken place. More precisely, we introduce a set of assumptions on the physical parameters characterizing the system, collectively described by a small parameter $\ve >0$, and we give an asymptotic expression for the solution of the Schr\"odinger equation up to order $\ve^2$ with an explicit control of the error.

\n
Roughly speaking, the result shows that the probability that more than one oscillator is excited is negligible and therefore the evolution of the system can be decomposed in only two rather different ''histories''. In the first one  the oscillators remain in their ground state and the test particle is described by the free evolution of a (slightly deformed)  spherical wave. The second history consists of a sum of $N$ terms where in each term  there is only one excited  oscillator. Here the test particle is correspondingly described by the free evolution of a wave packet, well concentrated in position and momentum, emerging from the excited oscillator with an average momentum parallel to the line joining the origin with the oscillator. 

\n
Such wave packet represents a semiclassical state for the test particle, propagating along the corresponding  classical trajectory (the straight track observed in the cloud chamber).

\n
We stress that the main result of our analysis  is to show how such a semiclassical state can be produced, starting from the original  spherical wave,  as a result of an interaction with the environment.  Moreover we emphasize that the interaction with the environment is   entirely described in terms of  the Schr\"odinger dynamics without any recourse to wave packet collapse rule.  

\n
The paper is organized as follows.

\n
In section 2 we give a detailed description of the model and we formulate our main result (theorem \ref{main}). 
In section 3 we describe the line of the proof and we formulate some intermediate results required to conclude the proof of theorem \ref{main}. 
In sections 4, 5, 6, 7, 8 we give a proof of the above mentioned intermediate results. 

\n
For the convenience of the reader, we collect here  some of the most used  notation  in the paper.

\n
- $\bs{x}=(x_1,x_2,x_3)$ is a vector in $\erre^3$, $|\bs{x}|$  the  euclidean  norm and $\hat{\bs{x}}= \f{\bs{x}}{|\bs{x}|}$  the corresponding unit vector. The scalar product in $\erre^3$ is $\bs{x}\cdot \bs{y}$.

\n
- $\un = (n_1,n_2,n_3)$ is a vector in $\enne^3$ and, with an abuse of notation, $|n| = n_1+n_2+n_3$.

\n
- $\langle \bo x \rangle = (1 + |\bo x|^2)^{1/2}$, $\bo x \in \erre^3$ and $\langle y \rangle = (1+y^2)^{1/2}$, $y \in \erre$.

\n
- $\|\cdot\|$, $(\cdot, \cdot)$ are the norm and the scalar product in $L^2(\erre^{3(N +1)})$ respectively.

\n
- $\|\cdot\|_{L^p}$, $p>0$, is the norm in $L^p (\erre^3)$ and $\langle \cdot,\cdot \rangle$ the scalar product in $L^2 (\erre^3)$.

\n
- The derivative of a function $f$ defined in $\erre^3$ is denoted by
\[ D^{\underline{\alpha}} f (\bo x)= \f{\partial^{|\alpha|} }{\partial x_1^{\alpha_1} \partial x_2^{\alpha_2} \partial x_3^{\alpha_3} } f(\bo x) 
\]
 for $\underline{\alpha} \in \enne^3$ and $|\alpha|= \alpha_1 + \alpha_2 + \alpha_3$.

\n
- $W^{n,1}_s (\erre^3)$, $n \in \enne$, $s>0$,  is the weighted Sobolev space equipped with the norm 
\[
\|f\|_{W^{n,1}_s}= \sum_{\underline{\alpha}, |\alpha| \leq n} \int \! d \bo x \, \langle \bo x \rangle^s |D^{\underline{\alpha}} f (\bo x) |
\]

\n
- $\tilde{f}$  is the Fourier transform of $f$.

\n
- We shall  use the following abbreviations for sums and products
\[   \sum_{\un} \equiv  \sum_{n_1, n_2, n_3 =1}^{\infty} \;\;\;\;\;\; \;\;\;\;\;\;\;\;\;\;\;\;\prod_{k, k\neq j}\equiv  \prod_{k=1, k\neq j}^{N}
\]
- During the proofs we shall often denote by $c$, $c_k$ a generic positive constant, possibly dependent on the integer $k$.


\section{Description of the model and main result}

\n
Let us consider a non relativistic quantum system made of $N+1$ spinless particles in dimension three, where one test particle has mass $M$ and the remaining $N$  particles with mass $m$ are bound by a harmonic potential of frequency $\omega$ around the equilibrium positions $\bs{a}_1,\ldots,\bs{a}_N$.  
In the model, the test particle plays the role of the $\alpha$-particle while the harmonic oscillators play the role of electrons in a simplified version of model-atoms with fixed nuclei. The interaction between the test particle and the $j$-th harmonic oscillator is described by a smooth two body potential $V$. We denote by $\bs{R}$ the position coordinate of the test particle and by $\bs{r}_1,\ldots,\bs{r}_N$ the position coordinates of the harmonic oscillators. Therefore the Hamiltonian of the system in $L^2(\erre^{3(N+1)})$ is given by
\begin{eqnarray}\label{hamiltoniana2osc}
H=H_0+\lambda \sum_{j=1}^NV_j
\ea
where
\ba
H_0=h_0+\sum_{j=1}^Nh_j,\;\;\;\;\;\;
h_0=-\frac{\hbar^2}{2M}\Delta_{\bs{R}},\;\;\;\; \;\;
h_j=- \frac{\hbar^2}{2m}\Delta_{\bs{r}_j}+\frac{1}{2}m\omega^2(\bs{r}_j-\bs{a}_j)^2
\end{eqnarray}
$\lambda>0$ is a coupling constant and $V_j$ is the multiplication operator by
\ba
V_j(\bs{R},\bs{r}_j) =V(\delta^{-1}(\bs{R} - \bs{r}_j)),\quad \;\;\;\;\;\; \delta >0
\ea
We are interested in the evolution of the system when the initial state is given in the product form
\begin{equation}
\Psi_0(\bs{R},\bs{r}_1,\ldots , \bs{r}_N)=\psi(\bs{R})\prod_{j=1}^N\varphi_{\uo,j}(\bs{r}_j)
\end{equation}
where $\psi(\bo R)$ is a spherical wave centered in the origin and   $\varphi_{\uo,j}(\bo r_j)$  is the ground state of the harmonic oscillator centered in $\bs{a}_j$. More precisely the spherical wave is given by
\begin{equation}\label{sphw}
\psi(\bs{R})=
\mathcal{N} \, f(\sigma^{-1}\bs{R}) \! \int_{S^2}  \!\! d\hat{\bo u}\, e^{i \frac{M v_0}{\hbar} \hat{\bs{u}}\cdot
\bs{R}}
\end{equation}
where $\mathcal{N}$ is a normalization constant, $\sigma, v_0>0$,  $f$  belongs to the Schwartz space $\mathcal{S}(\mathbb{R}^3)$ with $\|f\|_{L^2}=1$ and it is rotationally invariant. For the sake of concreteness  we choose
\be
f(\bs{x}) \, = \, \pi^{-3/4} \, e^{- \f{|\bs{x}|^2}{2}}
\ee
but it will be clear in the following that the result of our analysis is independent of the specific choice of $f$. 
Formula (\ref{sphw}) defines a spherical wave concentrated in position around the origin with an isotropic momentum $M v_0$. Moreover, for the eigenfunctions of the harmonic oscillators,  we denote $\un= (n_1,n_2,n_3) \in \enne^3$ and
\ba
&&\varphi_{\un,j} (\bs{r})=    \varphi_{\un} (\bs{r}-\bs{a}_j) = \gamma^{-3/2} \phi_{\un}(\gamma^{-1}(\bs{r}-\bs{a}_j)) ,  \;\;\;\;\;\;
\gamma=\sqrt{\frac{\hbar}{m\omega}}\\
&&\phi_{\un}(\bs{x}) \equiv \phi_{n_1} (x_1) \phi_{n_2}(x_2) \phi_{n_3}(x_3)
\ea
where $\phi_{n_k}$ is the Hermite function of order $n_k$. In particular the ground state corresponds to $\un=\uo=(0,0,0)$.

\n
In this generality, it is surely too hard to obtain a detailed description of the evolution of the system. We shall limit ourselves to consider an appropriate scaling limit and to derive an approximate evolution with an explicit control of the error. 
More precisely, we introduce a small parameter $\ve >0$ and fix
\be
\hbar =\ve^2 \;\;\;\;\;\; M=1 \;\;\;\;\;\; \sigma=\ve\;\;\;\;\;\;m=\ve \;\;\;\;\;\; \omega = \ve^{-1}\;\;\;\;\;\; \delta=\ve \;\;\;\;\;\; \lambda=\ve^2
\ee
Under this scaling the Hamiltonian becomes
\begin{eqnarray}\label{hamilt_2oscillatori}
H^{\varepsilon}=H^{\varepsilon}_0+\varepsilon^2  \, V^{\ve}
\ea
where
\ba
H^{\varepsilon}_0=h_0^{\varepsilon}+\sum_{j=1}^Nh_j^{\varepsilon},\;\;\;\;\;\;
 h_0^{\varepsilon}=-\frac{\varepsilon^4}{2}\Delta _{\bs{R}}, \;\;\;\;\;\;
 h_j^{\varepsilon}=\varepsilon^{-1} \left[-\frac{\varepsilon^4}{2}\Delta_{\bs{r}_j}+\frac{1}{2}(\bs{r}_j- \bs{a}_j)^2\right]
 \ea
 and
 \ba
V^{\ve}= \sum_{j=1}^N V^{\varepsilon}_j  , \;\;\;\;\;\;\;\;     V^{\varepsilon}_j(\bs{R},\bs{r}_j)=V\left(\ve^{-1}(\bs{R}-\bs{r}_j)  \right)
\end{eqnarray}
The rescaled initial state of the system is
\begin{eqnarray}\label{datoinizialerisc2}
\Psi_0^{\varepsilon}(\bs{R},\bs{r}_1,\ldots , \bs{r}_N)=\psi^{\varepsilon}(\bs{R})\prod_{j=1}^N\varphi_{\uo,j}^{\varepsilon}(\bs{r}_j)\\
\psi^{\varepsilon}(\bs{R})=\frac{\mathcal{N}_{\varepsilon}}{\varepsilon^{5/2}}\, f \left(\ve^{-1} \bs{R} \right)
\int_{S^2}\!\!  d\hat{\bs{u}}\,e^{\frac{i}{\varepsilon ^2}v_0 \hat{\bs{u}}\cdot
\bs{R}}\\
\varphi_{\un,j}^{\varepsilon}(\bs{r}_j)=\frac{1}{\ve^{3/2}}\phi_{\un} \left(\ve^{-1}(\bs{r}_j-\bs{a}_j)  \right)\;\;\;\;\;\;
\un \in \mathbb{N}^3
\end{eqnarray}
By a direct computation one sees that the normalization constant $\mathcal N_{\ve}$ for $\ve \rightarrow 0$ reduces to
\be
\mathcal N_0=\f{v_0}{4  \pi}
\ee
We notice that under this scaling the energy levels of each harmonic oscillator are
\be
E_{\un}^{\ve}= \ve \left( |n| + \f{3}{2} \right), \;\;\;\;\;\;\;\; |n|=n_1+n_2+n_3
\ee
Furthermore we introduce the following assumptions:

\vs
\n
(A) \hspace{0.3cm} {\em The Fourier transform $\widetilde{V}$ of the interaction potential $V$ belongs to the weighted Sobolev 

\n
$\;\;\;\;\;\; \;\;\;\,  $ space $W^{4,1}_4 (\erre^3)$.}

\n
(B)  \hspace{0.3cm} {\em The positions of the oscillators $\bs{a}_1,\ldots, \bs{a}_N$ satisfy the two conditions:

\n
$\;\;\;\;\;\; \;\;\;\; \, |\bs{a}_1|<|\bs{a}_2|<\ldots <|\bs{a}_N| \;  $ and  $\;\; \bs{a}_i\cdot \bs{a}_j \neq |\bs{a}_i | |\bs{a}_j | \,, \,\,\,\,\,\, i\neq j $.}

\vs

\n

\n
We notice that under assumption (A)  the Hamiltonian $H^{\ve}$ is surely self-adjoint and bounded from below (in fact much less  is required). In particular this implies that the unitary evolution of the system is well defined for any time $t$. Assumption (A) is also crucial  for our specific method of proof, based on repeated integrations by parts in highly oscillatory integrals. Furthermore 
assumption (B) is a technical ingredient useful for the proof. 

\n
We stress that our model is completely defined by the Hamiltonian (\ref{hamilt_2oscillatori}), the initial state (\ref{datoinizialerisc2}) and the assumptions (A), (B). Before approaching the evolution problem,  we briefly comment on the physical meaning of our scaling for $\ve  \rightarrow 0$.

\n
We  first observe that the dimensionless quantity
\be
\f{\hbar}{Mv_0 \sigma}
\ee
is of oder $\ve$, which means that the wavelength  $\f{\hbar}{Mv_0}$ associated to the $\alpha$-particle is much smaller than the spatial localization $\sigma$ (high momentum regime). 
Analogously for any $j=1,\ldots , N$ the quantities
\be
\f{\sigma}{|\bs{a}_j|}, \;\;\;\;\;\; \f{\gamma}{|\bs{a}_j|}, \;\;\;\;\;\; \f{\delta}{|\bs{a}_j|}
\ee
are of order $\ve$, i.e. the spatial localization of the test particle, the ''diameter'' of the oscillators and the range of the interaction are much smaller than the macroscopic distance $|\bs{a}_j|$. 

\n
Moreover we notice that the energy levels of the oscillators $\hbar \omega (|n|+3/2)$ are of order $\ve$ and the coupling constant $\lambda$ is of order $\ve^2$ while the kinetic energy of the test particle is of order one for $\ve \rightarrow 0$.  
This guarantees that  the energy loss for the test particle due to the interaction with an oscillator is very small (quasi-elastic regime) and the perturbative approach can be used. 

\n
Finally it is interesting to compare the characteristic times of the system.  In particular we define the classical flight times for $ j=1, \ldots, N$ as the time spent by a classical particle, starting from the origin with velocity $v_0$, to reach the oscillator in $\bs{a}_j$
\be\label{tau}
\tau_j = \f{|\bs{a}_j|}{v_0}
\ee
the ''period'' of the oscillators
\be
T_o = \f{2 \pi }{\omega}
\ee
and the transit time, i.e. the time spent by  the test particle to travel the diameter of an  oscillator
\be
T_t = \f{\gamma}{v_0}
\ee
It turns out that
\be
\f{T_t}{T_o} =O(1)
\ee
i.e. the test particle can ''see'' the internal structure of the oscillators and
\be
\f{T_t}{\tau_j} = O(\ve)
\ee
which implies that  $\tau_j$ can be reasonably identified as the collision time of the test particle with the oscillator in $\bs{a}_j$.

\n
We are now ready to study the solution of the Schr\"odinger equation   with initial datum $\Psi^{\ve}_0$
\be
\mue(t)\Psi^{\ve}_0 \,, \;\;\;\;\;\;\;\;\;  \mue (t)= e^{-i \f{t}{\ve^2} H^{\ve} }
\ee
for $t>\tau_N$, i.e. for a fixed time $t$ sufficiently large so that all collisions of the test particle with the oscillators have taken place.
In particular we shall perform a perturbative analysis computing the first correction for $\ve \rightarrow 0$ to the free evolution of the system
\be
\mul (t)\Psi^{\ve}_0 \,, \;\;\;\;\;\;\;\;\;      \mul (t)= e^{-i \f{t}{\ve^2} H_0^{\ve} }
\ee

\n
In order to formulate our main result, it is convenient to introduce the following definition.
\begin{definizione}\label{depac}
Let $P_j^{\ve} = P_j^{\ve}(\bs{R},\bs{r}_1,\ldots,\bs{r}_N)$ be the following function
\be
P_j^{\ve}(\bs{R},\bs{r}_1,\ldots,\bs{r}_N)= \sum_{\un} P_{\un ,j}^{\ve} (\bs{R}) \, \varphi^{\ve}_{\un,j} (\bs{r}_j) \prod_{k, k \neq j}  \varphi^{\ve}_{0 ,k} (\bs{r}_k)
\ee
where $P_{\un,j}^{\ve}$ is the wave packet for the test particle given by
\ba\label{Pnj}
&&P_{\un ,j}^{\ve} (\bs{R}) = \f{C_{\un,j}^{\ve}}{\ve^{3/2}}\,  \mathcal A_{\un,j}\! \left(  \f{ \! \bs{R} -\! (\hat{\bs{a}}_j  \! \cdot \!  \bs{R} ) \hat{\bs{a}}_j}{\ve} \!  \right) \, e^{- \f{\left( \hat{\bs{a}}_j \cdot \bs{R} - \mathcal Z^{\ve}_{\un ,j} \right)^2}{2 \ve^2}  
\,+  \,i\, \f{ \mathcal V_{\un}^{\ve}}{\ve^2} \;\hat{\bs{a}}_j \cdot \bs{R} } \\
&&C_{\un,j}^{\ve}=  -  \, \f{ i  8 \pi^{9/4} \mathcal N_{\ve} }{v_0^3 \tau_j^2} \, e^{\f{i}{\ve} |n| \tau_j + i \f{|n|^2 \tau_j}{2 v_0^2}}   \label{Cnj}\\
&&\mathcal A_{\un,j}(\bs{y})   =   e^{-i \f{|\bs{y}|^2}{2 \tau_j}} \left[ \widetilde{V} \cdot \widetilde{(\phi_{\un} \phi_{\uo})} \right] (-\tau_j^{-1} \bs{y} +|n| v_0 \hat{\bs{a}}_j )  \label{cala}\\
&&\mathcal Z^{\ve}_{\un, j} = \ve\, \f{|n| \tau_j}{v_0}\\
&&\mathcal V_{\un}^{\ve}=  v_0 - \ve \,\f{ |n|}{v_0}
\ea
\end{definizione}

\n
We underline that the wave packet $P^{\ve}_{\un,j}$ plays  a crucial role in our analysis. It is written as the product of two different wave packets. The first one is  a function of the two dimensional vector $ \bs{R} -\! (\hat{\bs{a}}_j  \! \cdot \!  \bs{R} ) \hat{\bs{a}}_j$, orthogonal to the direction of the $j$-th oscillator $\hat{\bs{a}}_j$, and it is well concentrated in position and  momentum around the origin for $\ve \rightarrow 0$.  The second one is a one-dimensional gaussian wave packet in the variable $\hat{\bs{a}}_j \cdot \bs{R}$, i.e. a coordinate along the direction of the $j$-th oscillator $\hat{\bs{a}}_j$. For $\ve \rightarrow 0$ such wave packet is well concentrated in position around $\mathcal Z^{\ve}_{\un,j}  $ and in momentum around $\mathcal V_{\un}^{\ve}$. 

\n
With the notation introduced in definition \ref{depac}, we state our main result.
\begin{teorema}\label{main}
Let us assume  $(A), (B)$. Then for any $t >\tau_N$ there exists $C(t)>0$ such that 
\be\label{teo}
\mue (t) \Psi_0^{\ve} \,= \,\mul (t) \Psi_0^{\ve}\, +\, \ve^2 \sum_{j=1}^N \mul (t) P_j^{\ve} \,+\, \mathcal E ^{\ve} (t)
\ee
where
\be\label{er}
\| \mathcal E^{\ve} (t) \| \leq C(t) \, \| \widetilde{V} \|_{W^{4,1}_4} \, \ve^3
\ee
\end{teorema}

\vs
\n
Let us conclude this section with few comments. Theorem \ref{main} provides the required approximate dynamics of the system for $\ve$ small. Using the expressions for  the free propagator $\mathcal U_0 (t)$, the initial state $\Psi_0^{\ve}$ and the functions $P_j^{\ve}$,  formula (\ref{teo}) can be rewritten as
\ba\label{teo1}
&& \left( \mue (t) \Psi_0^{\ve} \right)(\bs{R},\bs{r}_1, \ldots , \bs{r}_N)\nonumber\\
&&= e^{- i \f{t}{\ve^2} N E^{\ve}_{\uo} } \left[ \left(  e^{- i \f{t}{\ve^2} h_0^{\ve} } \psi^{\ve} \right) (\bs{R}) + \ve^2 \sum_{j=1}^N \left( e^{ -i \f{t}{\ve^2} h_0^{\ve} } P_{\uo , j}^{\ve} \right) (\bs{R}) \right] \prod_{k=1}^N \varphi^{\ve}_{\uo ,k} (\bs{r}_k) \nonumber\\
&&+ \; \ve^2 \, e^{-i \f{t}{\ve^2} (N-1) E^{\ve}_{\uo} } \sum_{j=1}^N \sum_{\un \neq \uo} e^{ -i \f{t}{\ve^2} E^{\ve}_{\un} }  \left( e^{-i \f{t}{\ve^2} h_0^{\ve} } P_{\un ,j}^{\ve} \right)\!(\bs{R})\,  \varphi_{\un , j} (\bs{r}_j) \prod_{k, k \neq j} \varphi_{\uo , k}(\bs{r}_k) \nonumber\\
&&+\;  \mathcal E^{\ve}(t)
\ea

\n
In the above formula the approximate wave function of the system has been written as the sum of two terms (or histories). In the first one  all oscillators remain in their ground state and the test particle is described by the free evolution of the spherical wave slightly deformed  by the free evolution of the small wave packets $P_{\uo ,j}^{\ve}$, $j=1,\ldots ,N$ emerging from each oscillator. The second term  is  a sum over $j$, $j=1, \ldots ,N$, where, in each term of the sum, one has only one oscillator in an excited state (say the $j$-th oscillator) and correspondingly the test particle  described by the free evolution of the small  wave packet $P_{\un ,j}^{\ve}$, $\un \neq \uo$, emerging from the excited oscillator. We recall that the  small wave packet  emerging from the $j$-th oscillator is well concentrated in position and momentum. Therefore, under free evolution, it propagates along the direction $\hat{\bs{a}}_j$ with momentum   $\mathcal V^{\ve}_{\un}$. We also notice that, for each $j$,  the wave packet $P_{\uo ,j}^{\ve}$ is produced by an elastic collision between  the test particle (with momentum $v_0$) and the $j$-th oscillator and therefore its momentum is unaffected ($\mathcal V_{\uo}^{\ve}= v_0$). On the other hand, the wave packet $P_{\un ,j}^{\ve}$, $\un \neq \uo$,  is produced by an inelastic collision with energy loss $\Delta E= \ve |n|$. Therefore, by the conservation of energy, its momentum is correctly given by $\mathcal V^{\ve}_{\un}$ at first order in $\ve$.

\n

\n
Finally, in order to  make more transparent the structure of the wave packet emerging from an excited oscillator, we give an explicit computation in a specific case. In particular we consider the case of the oscillator in $\hat{\bo a}_1$ excited to the state labeled by $\un 
$ where, without loss of generality, we choose  $
\hat{\bo a}_1 = (0,0,1)$.    Moreover 
we denote $\bo R=(X,Y,Z)$ and use the shorthand notation $C^{\ve}=C^{\ve}_{\underline{n},1}$, $\mathcal A (X,Y) = \mathcal A_{\underline{n},1} (X,Y,0)$, $\mathcal Z^{\ve} = \mathcal Z^{\ve}_{\underline{n}, 1}$, $\mathcal V^{\ve}= \mathcal V^{\ve}_{\underline{n}}$. Then  the wave packet $P^{\ve}_{\underline{n},1} $ is factorized into a function of the variables $(X,Y)$, depending on the product  $\widetilde{V}\!\cdot \widetilde{ \phi_{\un} \phi_{\uo}}$, and a gaussian  in the variable $Z$ 
\be\label{P1}
P^{\ve}_{\underline{n},1}(X,Y,Z) = \f{C^{\ve} }{\ve^{3/2}} \; \mathcal A(\ve^{-1}X, \ve^{-1} Y) \; e^{- \f{ (Z - \mathcal Z^{\ve} )^2}{2 \ve^2}  \,+\, i \, \f{\mathcal V^{\ve}  }{\ve^2} Z }
\ee
with
\be\label{norP1}
\|P^{\ve}_{\underline{n},1} \|_{L^2}^2 = \sqrt{\pi} \; |C^{\ve}|^2 \int\!\! dX dY \, |\mathcal A (X,Y) |^2
\ee
From (\ref{P1}), (\ref{norP1}) one easily computes mean value and standard deviation for the position. In particular one finds that $\langle X \rangle$, $\langle Y \rangle$, $\langle Z \rangle$, $\Delta X$, $\Delta Y$, $\Delta Z$ are all  $O(\ve)$. In the Fourier space one has
\be\label{FP1}
\widetilde{P}^{\ve}_{\underline{n},1} (K_x, K_y, K_z) = \ve^{3/2} C^{\ve} \widetilde{\mathcal A}(\ve K_x, \ve K_y) \; e^{ - \f{\ve^2}{2} \left( K_z - \mathcal V^{\ve} / \ve^2 \right)^2  - i \, \mathcal Z^{\ve} K_z + \f{i}{\ve^2} \mathcal Z^{\ve} \mathcal V^{\ve} } 
\ee
Therefore for the mean value of the momentum one has $\langle P_x \rangle =O(\ve)$, $\langle P_y \rangle =O(\ve)$, $\langle P_z \rangle = \mathcal V^{\ve} = v_0 +O(\ve)$,  while for the standard deviation one finds that $\Delta P_x$, $\Delta P_y$, $\Delta P_z$ are all $O(\ve)$.

\n
The free evolution of the wave packet  is also factorized into a free evolution of its part dependent on the variables $(X,Y)$ and the free evolution of the gaussian part dependent on $Z$. Therefore, from the above computations, one can conclude that,  for $t$ not too large, the wave packet remains well concentrated in position and momentum and it propagates along the $Z$-axis with momentum approximately given by $v_0$. This means that for $t=\tau_1$ the wave packet is concentrated around the oscillator in $\bo a_1$ and for $t>\tau_1$ it will continues its propagation along the $Z$-axis.


\section{Line of the proof}

\n
The proof of theorem \ref{main} requires some intermediate results. In this section we describe the line of reasoning, we state without proof such intermediate results and we finally conclude with the proof of the theorem. The proof of the intermediate results will be given in sections 4, 5, 6, 7, 8. 
We start with Duhamel's formula to represent the solution of the Schr\"odinger equation
\be
\mue (t) \Psi_0^{\ve} = \mul (t) \Psi_0^{\ve} - i \int_0^t \!\! ds\; \mue (t-s) \,V^{\ve} \, \mul (s) \Psi_0^{\ve}
\ee
Iterating twice  we obtain
\ba\label{repd}
&&\mue (t) \Psi_0^{\ve} =   \mul (t)\bigg[  \Psi_0^{\ve}  +  I^{\ve}(t) \Psi_0^{\ve} \bigg] +  \mathcal R^{\ve}(t) =  \mul (t)\left[  \Psi_0^{\ve}  + \sum_{j=1}^N I_j^{\ve}(t) \Psi_0^{\ve} \right] +  \mathcal R^{\ve}(t)
\ea
where we have denoted
\ba
&&I^{\ve}(t)=  - i \int_0^t \!\! ds\; \mul (-s) \,V^{\ve} \,\mul (s) \label{Ive} \\
&&I_j^{\ve}(t)=  - i \int_0^t \!\! ds\; \mul (-s) \,V_j^{\ve}\, \mul (s) \label{Igei} \\
&&\mathcal R^{\ve} (t)= \mul (t) J^{\ve}(t) \Psi_0^{\ve} -  i \int_0^t \!\!ds\; \mue (t-s) \,V^{\ve} \,\mul (s) J^{\ve}(s) \Psi_0^{\ve}  \label{Rve} \\
&&J^{\ve} (t)=- \int_0^t \!\! ds \! \!\int_0^s \!\! d\sigma\; \mul (-s)\, V^{\ve} \,\mul (s) \mul (-\sigma)\, V^{\ve} \,\mul(\sigma) \label{Jve}
\ea
In order to isolate the dominant term in $I^{\ve}_j(t) \Psi_0^{\ve}$ for $\ve \rightarrow 0$ it is convenient to introduce some further notation.
Let us consider 
the rotation matrix $\mathcal R_j$, $j=1,\ldots ,N$, defined by the condition
\be\label{mr}
\mathcal R_j \hat{\bs{a}}_j = (0,0,1)
\ee
and the angle
\be
\theta_0= \f{1}{2} \min_{j,k, j\neq k} \big\{ \theta_{jk} , \pi \big\}\,, \;\;\;\;\;\; \;\; \theta_{jk}= \cos^{-1} \big( \hat{\bs{a}}_j \cdot \hat{\bs{a}}_k \big)
\ee
Furthermore we introduce the following two portions of the unit sphere $S^2$ around the unit vectors $(0,0,1)$ and $\hat{\bs{a}}_j$ respectively
\ba
&&\mathcal C_0 = \left\{ \hat{\bs{u}} \in S^2 \;|\; \hat{u}_1^2 + \hat{u}_2^2 < \sin^2 \theta_0 \right\}\\
&&\mathcal C_j = \mathcal R_j \mathcal C_0 = \left\{ \hat{\bs{u}} \in S^2 \;|\; \left( \mathcal R_j \hat{\bs{u}}\right)_1^2 + \left( \mathcal R_j \hat{\bs{u}}\right)_2^2 < \sin^2 \theta_0 \right\}
\ea
We notice that $\mathcal C_j \cap \mathcal C_k = \emptyset$, for $j\neq k$; moreover $\hat{\bs{a}}_j \in \mathcal C_j$ and $\hat{\bs{a}}_k \notin \mathcal C_j$ for $j \neq k$. 
Exploiting the above notation, we can define the ''portion around $\hat{\bs{a}}_j$'' of the initial spherical wave
\be\label{psij}
\psi_j^{\varepsilon}(\bs{R})=\frac{\mathcal{N}_{\varepsilon}}{\varepsilon^{5/2}}\, \eta \left(\ve^{-1} \bs{R} \right)
\int_{\mathcal C_j}\!\!  d\hat{\bs{u}}\,e^{\frac{i}{\varepsilon ^2}v_0 \hat{\bs{u}}\cdot\bs{R}}
\ee
and correspondingly
\be\label{PJ}
\Psi_{0,j}^{\varepsilon}(\bs{R},\bs{r}_1,\cdots , \bs{r}_N)=\psi_j^{\varepsilon}(\bs{R})\prod_{k=1}^N\varphi_{\uo,k}^{\varepsilon}(\bs{r}_k)
\ee
The portion of spherical wave defined in (\ref{psij}) has the following property: its classical evolution will hit the oscillator in $\bs{a}_j$ and it will not hit the remaining oscillators in $\bs{a}_k$, for $k \neq j$. 
Taking into account definition (\ref{PJ}), we rewrite (\ref{repd}) as follows
\be\label{repd2}
\mue (t) \Psi^{\ve}_{0}=  \mul (t)\left[  \Psi_0^{\ve}  + \sum_{j=1}^N I_j^{\ve}(t) \Psi_{0,j}^{\ve} \right] +   \sum_{j=1}^N   \mul(t)  I_j^{\ve}(t) \left( \Psi_0^{\ve} -  \Psi_{0,j}^{\ve} \right)   +    \mathcal R^{\ve}(t)
\ee
The problem is now reduced to the analysis of the r.h.s. of (\ref{repd2}) for $\ve \rightarrow 0$. In particular, in order to prove theorem \ref{main} we have to show that  $I_j^{\ve}(t) \Psi_{0,j}^{\ve} = \ve^2 P_j^{\ve} +O(\ve^3)$ and also that the last two terms of (\ref{repd2}) are $O(\ve^3)$.

\n
The first step is to obtain a good representation of the relevant objects  $I_j^{\ve}(t) \Psi_{0,j}^{\ve}$, $\,I_j^{\ve}(t) ( \Psi_{0}^{\ve} - \Psi_{0,j}^{\ve})\,$ and $\,J^{\ve}(t) \Psi_{0}^{\ve}$. This is done in section 4,   where for each object we perform an expansion in  series of the eigenfunctions of the harmonic oscillators and the coefficients of the expansion are written as highly oscillatory integrals. Such representation formulas allow to exploit stationary and non stationary phase methods to characterize the asymptotic behavior of each object  for $\ve \rightarrow 0$ (\cite{f}, \cite{bh}). In the case of $I_j^{\ve}(t) \Psi_{0,j}^{\ve}$, it turns out that the phase in the oscillatory integral has exactly one, non degenerate, critical point in the integration region if $t>\tau_j$. This is the crucial ingredient to prove
\begin{prop}\label{prop2}
Let us assume $(A), (B)$. Then for any $t>\tau_j$  there exists $C^1(t)>0$ such that 
\be\label{deco}
I^{\ve}_j (t) \Psi^{\ve}_{0,j} \,=\,  \ve^2 P_j^{\ve} \,+\,  Q_j^{\ve} (t)
\ee
where $P_j^{\ve}$ is given in definition \ref{depac} and
\be\label{stiQ}
\|Q_j^{\ve} (t)\| \leq C^{1} (t) \, \|\widetilde{V} \|_{W^{4,1}_4} \, \ve^3
\ee
\end{prop}
\n
We shall prove  decomposition (\ref{deco}), with an explicit expression for $Q_j^{\ve} (t)$, in section 5, while the estimate (\ref{stiQ}) is proved in section 6.

\n
For the estimate of $\,I_j^{\ve}(t) ( \Psi_{0}^{\ve} - \Psi_{0,j}^{\ve})\,$ one exploits the fact that the corresponding phase in the oscillatory integral  does not have critical points in the integration region. Therefore, by non stationary phase methods, in section 7 we prove the following proposition.

\begin{prop}\label{prop3} Let us assume $\widetilde{V} \!  \in  W^{k,1}_k (\erre^3)$, $k \in \enne$, and $(B)$. Then for any $t>0$ there exists $C^2_k (t)>0$  such that
\be\label{Ijns}
 \|  I_j^{\ve}(t) \left( \Psi_0^{\ve} -  \Psi_{0,j}^{\ve} \right) \| \leq C^{2}_k (t) \,\|\widetilde{V} \|_{W^{k,1}_k} \,  \ve^{k-1}
 \ee
 \end{prop}

\n
The estimate of the rest $\mathcal R^{\ve} (t)$ is reduced to the estimate of $J^{\ve} (t) \Psi^{\ve}_0$.  Exploiting the representation given in section 4, in section 8 we prove

\begin{prop}\label{prop4} Let us assume $(A), (B)$. Then for  any $t>0$  there exists $C^3(t)>0$ such that 
\be\label{stR}
\|\mathcal R^{\ve}(t) \| \leq \left(1 + t \|V\|_{L^{\infty}} \right) \sup_{s \leq t} \|J^{\ve}(s) \Psi_0^{\ve} \|  \leq  C^3 (t)\, \left(1 + t \|V\|_{L^{\infty}} \right)    \|\widetilde{V} \|_{W^{4,1}_4} \, \ve^3
\ee
\end{prop}

\n
Collecting together the results stated in the above propositions 
we are in position to prove our main result.

\vs

\n
{\em Proof of theorem \ref{main}}. 
The proof follows from  formula (\ref{repd2}), proposition \ref{prop2},    estimate (\ref{Ijns}) for $k=4$ and estimate (\ref{stR}).
\hfill $\Box$


\section{Representation formulas}

In this section we derive useful representation formulas for the relevant objects  $I_j^{\ve}(t) \Psi_{0,j}^{\ve}$, $\,I_j^{\ve}(t) ( \Psi_{0}^{\ve} - \Psi_{0,j}^{\ve})\,$ and $\,J^{\ve}(t) \Psi_{0}^{\ve}$.
In particular we perform an expansion in eigenfunctions of the harmonic oscillators and we represent the corresponding Fourier coefficients as highly oscillatory integrals. Indeed, in terms of the rescaled variable
\be
\bo x=\frac{\bo R}{\ve}
\ee
we have
\begin{prop}\label{prop1}
For any $\ve> 0$, $\un\in\mathbb{N}^3$  and $j=1, \ldots , N$ the following representation formulas hold
\ba\label{Ij}
&&\left( I^{\ve}_j (t) \Psi^{\ve}_{0,j} \right)(\ve \bo x,  \bo r_1,\ldots, \bo r_N)
=   \sum_{\un}  \mathcal I^{\ve}_{j, \un}(t,\bo x)  \varphi^{\ve}_{\un,j}(\bo r_j)  \prod_{k, k\neq j}\varphi^{\ve}_{\uo,k} (\bo r_k),
\ea
where
\ba\label{cIj}
&&\mathcal I^{\ve}_{j, \un}(t,\bo x) =  - \f{i\, \mathcal N_{\ve}}{\ve^{5/2}}  \int_0^t \!\!\!ds \!\! \int\!\!d\bo \xi \!\! \int_{\mathcal C_j} \!\!\! d\hat{\bo u} \; F_{\un} (\bo \xi,s;\bo x) \; e^{ \f{i}{\ve} \Phi_{\un,j} (\bo \xi, s, \hat{\bo u};\bo x)}\\
&& F_{\un} (\bo \xi,s; \bo x) = e^{i \bo \xi \cdot \bo x +i \f{s}{2} \bo \xi^2} g_{\un,\uo}(\bo \xi)\, f(\bo x+s \,\bo \xi) \label{Fn}   \\
&&g_{\un, \um}(\bs{\xi}) = \widetilde{\phi_{\un} \phi_{\um}} (\bs{\xi})  \tilde{V} (\bs{\xi}) \label{gnm}\\
&&\Phi_{\un,j} (\bo \xi, s, \hat{\bo u};x)= - \bo \xi \cdot \bo a_j + v_0\, \hat{\bo u} \cdot (\bo x+ s\, \bo \xi) +|n| s  \label{phinj}
\ea
and
\ba\label{Ijnon-staz}
&&\left( I^{\ve}_j (t) (\Psi^{\ve}_0-\Psi^{\ve}_{0,j})  \right)(\ve \bo x,  \bo r_1,\ldots, \bo r_N)
=   \sum_{\un}  \mathcal T^{\ve}_{j, \un}(t,\bo x)  \varphi^{\ve}_{\un,j}(\bo r_j)  \prod_{i, i\neq j}\varphi^{\ve}_{\uo,i} (\bo r_i),
\ea
where $\mathcal T^{\ve}_{j, \un}$  differs from $\mathcal I^{\ve}_{j, \un}$ only for the integration region, i.e. $\mathcal C_j$ is replaced by $S^2\backslash \mathcal C_j$.
\end{prop}

\begin{proof}     We shall first give the representation formula  for 	$ I^{\ve}_j (t) \Psi^{\ve}_0 $.
From the definition (\ref{Igei}) of $I_j^{\ve}(t)$ we have that
\ba
&& \int \!\! d \bs{r}_1 \ldots d \bs{r}_N \, \varphi^{\ve}_{\un_{1},1} (\bs{r}_1)    \cdots  \varphi^{\ve}_{\un_{N},N} (\bs{r}_N) \left( I_j^{\ve} (t) \Psi_0^{\ve} \right) ( \bo R,  \bo r_1,\ldots, \bo r_N)    \nonumber\\
&&= -i \int_0^t \!\!\! ds\, e^{i \f{s}{\ve^2} E_{\un_1}^{\ve} +\cdots + i \f{s}{\ve^2} E_{\un_N}^{\ve} - i \f{s}{\ve^2} E_{\uo}^{\ve} - \cdots  -  i \f{s}{\ve^2} E_{\uo}^{\ve} }
\int \!\! d \bs{r}_1 \ldots d \bs{r}_N \, \varphi^{\ve}_{\un_{1},1} (\bs{r}_1)    \cdots  \varphi^{\ve}_{\un_{N},N} (\bs{r}_N) \nonumber\\
&&\cdot \left( e^{i \f{s}{\ve^2} h_0^{\ve} } V_j^{\ve} e^{- i \f{s}{\ve^2} h_0^{\ve} }  \psi^{\ve} \varphi^{\ve}_{\uo, j} \right) \!(\bs{R},\bs{r}_j) \prod_{k, k\neq j} \varphi_{\uo, k}(\bs{r}_k) \nonumber\\
&&=-i \prod_{k, k\neq j} \delta_{\un_k , \uo} \int_0^t \!\!\! ds\, e^{i \f{s}{\ve} |n_j|} \int \!\! d \bs{r}_j \, \varphi^{\ve}_{\un_j, j} (\bs{r}_j) \left(  e^{i \f{s}{\ve^2} h_0^{\ve} } V_j^{\ve} e^{- i \f{s}{\ve^2} h_0^{\ve} }  \psi^{\ve} \varphi^{\ve}_{\uo, j} \right) \!(\bs{R},\bs{r}_j)
\ea
Therefore
\ba
&&\left( I^{\ve}_j (t) \Psi^{\ve}_0 \right)( \bo R,  \bo r_1,\ldots, \bo r_N)
=   \sum_{\un}  \mathcal I^{\ve}_{j, \un}(t,\bo R)  \varphi^{\ve}_{\un,j}(\bo r_j)  \prod_{k, k\neq j}\varphi^{\ve}_{\uo,k} (\bo r_k),\ea
where
\ba\label{1}
&&\mathcal I^{\ve}_{j, \un}(t,\bo R)= - i \!\!\int_0^t \!\! ds\,  e^{i \f{s}{\ve} |n|}   \left( e^{i \f{s}{\ve^2} h_0^{\ve} } \langle \varphi^{\ve}_{\un,j},
 V^{\ve}_j \varphi^{\ve}_{\uo,j} \rangle \,  e^{- i \f{s}{\ve^2} h_0^{\ve} } \psi^{\ve} \right) \!(\bs{R})
\ea
For the  scalar product in $L^2(\erre^3)$  appearing in the integrand in (\ref{1}) we have
\ba\label{sp2}
&&\langle\varphi^{\ve}_{\un,j},
 V^{\ve}_j \varphi^{\ve}_{\uo,j} \rangle(\bo R)= \frac{1}{\sqrt{\ve^3}}\!\!\int \!\!d\bo r_j\,\phi_{\un}(\ve^{-1}(\bo r_j\!-\bo a_j)) V(\ve^{-1}(\bo R-\bo r_j))\phi_{\uo}(\ve^{-1}(\bo r_j\!-\bo a_j))\nonumber\\
 &&= \int d\bo x\,
\phi_{\un}(\bo x)\phi_{\uo}(\bo x)V\!\left(\frac{\bo R-\bo a_j}{\ve}-\bo x\right)\nonumber   =\int d\bo \xi \,\widetilde{\phi_{\un}\phi_{\uo}}(\bo \xi)\tilde{V}(\bo \xi)e^{\frac{i}{\varepsilon}\bo \xi\cdot
(\bo R-\bo a_j)}\nonumber\\
&&\equiv \int d\bo \xi \,g_{\un,\uo}(\bo \xi)e^{\frac{i}{\varepsilon}\bo \xi\cdot
(\bo R-\bo a_j)}
 \ea
hence
\ba\label{Ijn2}
&&\mathcal I^{\ve}_{j, \un}(t,\bo R)= - i \!\!\int_0^t \!\! ds\,  e^{i \f{s}{\ve} |n|}\! \!  \int\!\! d\bs{\xi} \, g_{\un, \uo}(\bs{\xi}) e^{- \f{i}{\ve} \xi \cdot \bs{a}_j} \left( e^{i \f{s}{\ve^2} h_0^{\ve} } e^{\f{i}{\ve} \bs{\xi} \cdot (\cdot)} e^{- i \f{s}{\ve^2} h_0^{\ve} } \psi^{\ve} \right)\!(\bs{R})
\ea
where $e^{\f{i}{\ve} \bs{\xi} \cdot (\cdot)}$ denotes the multiplication operator $(e^{\f{i}{\ve} \bs{\xi} \cdot (\cdot)} g)(\bs{R})= e^{\f{i}{\ve} \bs{\xi} \cdot \bs{R}} g(\bs{R})$.
Furthermore for any $g \in L^2(\mathbb{R}^3)$ we have
 \begin{eqnarray}\label{azione}
&& \left(e^{\frac{i}{\ve^2}sh_0^{\ve}}
e^{\frac{i\bo\xi}{\ve}\cdot (\cdot)}
e^{-\frac{i}{\ve^2}sh_0^{\ve}}g \right)(\bo R) =
\frac{1}{\sqrt{(2\pi)^3}}\int d\bo k\, e^{i\bo k\cdot
\bo R}\left(e^{\frac{i}{\ve^2}sh_0^{\ve}}e^{\frac{i\bo \xi}{\ve}\cdot
(\cdot)
}e^{-\frac{1}{\ve^2}sh_0^{\ve}}g  \right)^\sim\!(\bo k)\nonumber\\
&&= \frac{1}{\sqrt{(2\pi)^3}}\int d\bo k\, e^{i\bo k\cdot
\bo R}e^{i\frac{s}{2}\bo k^2\ve^2}\left(e^{\frac{i\bo \xi}{\ve}\cdot
(\cdot)
}e^{-\frac{i}{\ve^2}sh_0^{\ve}}g  \right)^\sim\!(\bo k)\nonumber\\
&&= \frac{1}{\sqrt{(2\pi)^3}}\int d\bo k\, e^{i\bo k\cdot
\bo R}e^{i\frac{s}{2}\bo k^2\ve^2}
\left(e^{-\frac{i}{\ve^2}sh_0^{\ve}}g  \right)^\sim\!\left(\bo k-\frac{\bo \xi}{\ve}\right)\nonumber\\
&&=  \frac{1}{\sqrt{(2\pi)^3}}\int d\bo k\, e^{i\bo k\cdot
\bo R}e^{i\frac{s}{2}\bo k^2\varepsilon^2}e^{-i\frac{s\varepsilon^2}{2}(\bo k-\frac{\bo \xi}{\varepsilon})^2}
\tilde{g}\!\left(\bo k-\frac{\bo \xi}{\varepsilon}\right)\nonumber\\
&&= \frac{e^{-i\frac{s\bo \xi^2}{2}}}{\sqrt{(2\pi)^3}}\int
d\bo k\,e^{i\bo k\cdot(\bo R+s\varepsilon
\bo \xi)}\tilde{g}\!\left(\bo k-\frac{\bo \xi}{\varepsilon}\right)\nonumber\\
&&= e^{is\frac{\bo \xi^2}{2}}e^{i\frac{\bo \xi}{\varepsilon}\cdot
\bo R} g(\bo R+\varepsilon s \bo \xi) \label{poc}
\end{eqnarray}
Using (\ref{poc}) in (\ref{Ijn2}) we have
\ba\label{Ijn3}
&&\mathcal I^{\ve}_{j, \un}(t,\bo R)= - i \!\!\int_0^t \!\! ds \!  \!  \int\!\! d\bs{\xi} \, g_{\un, \uo}(\bs{\xi}) \,  e^{i \left( \f{s}{\ve} |n|  - \f{ \xi}{\ve} \cdot \bs{a}_j  +s \f{\bs{\xi}^2}{2} + \f{\bs{\xi}}{\ve} \cdot \bs{R} \right) } \psi^{\ve}(\bs{R} + \ve s \bs{\xi} )
\ea
Substituting the explicit expression of $\psi^{\ve}$ in (\ref{Ijn3}) one obtains the required representation of $I_j^{\ve}(t) \Psi^{\ve}_0$. The representation for $I_j^{\ve}(t) \Psi^{\ve}_{0,j}$ and  $I_j^{\ve}(t) (\Psi^{\ve}_0 - \Psi_{0,j}^{\ve}) $ are obtained replacing in (\ref{Ijn3}) $\psi^{\ve}$ with $\psi^{\ve}_j$ and $\psi^{\ve}- \psi^{\ve}_j$ respectively.
\end{proof}

\begin{prop}
For any $\ve> 0$, $\un,\um\in\mathbb{N}^3$  and $k,l=1,\ldots , N$ we have the following representation formula
\ba\label{J}
&&J^{\ve} (t) = \sum_{\substack{k,l=1\\k\neq l}}^NJ^{\ve}_{k,l} (t)    +        \sum_{k=1}^NJ^{\ve}_{k} (t)    \\
&&J^{\ve}_{k,l}(t)=- \int_0^t \!\! ds \! \!\int_0^s \!\! d\sigma\; \mul (-s) V^{\ve}_k \mul (s) \mul (-\sigma) V^{\ve}_l \mul(\sigma)\,, \;\;\;\;\;\;  J^{\ve}_{k}(t) \equiv J^{\ve}_{k,k}(t) \\
&&\left( J^{\ve}_{k} (t) \Psi^{\ve}_0  \right)(\ve \bo x,  \bo r_1,\ldots, \bo r_N)
=   \sum_{\un,\um}  \mathcal J^{\ve,k}_{ \un,\um}(t,\bo x)  \varphi^{\ve}_{\un,k}(\bo r_k) \prod_{i, i\neq k}\varphi^{\ve}_{\uo,i} (\bo r_i) 
\label{Jk} \\
&&\left( J^{\ve}_{k,l} (t) \Psi^{\ve}_0  \right)(\ve \bo x,  \bo r_1,\ldots, \bo r_N)
=   \sum_{\un,\um}  \mathcal J^{\ve,k,l}_{ \un,\um}(t,\bo x)  \varphi^{\ve}_{\un,k}(\bo r_k)\varphi^{\ve}_{\um,l}(\bo r_l)  \prod_{i,  i\neq k,l}\varphi^{\ve}_{\uo,i} (\bo r_i)
\label{Jkl}
\ea
where
\ba
&&\mathcal J^{\ve,k}_{\un,\um}(t,\bo x) =  - \f{ \mathcal N_{\ve}}{\ve^{5/2}}  \int_0^t \!\! ds \! \!\int_0^s \!\! d\sigma\;\!\! \int\!\!d\bo \xi \!\!\int\!\!d\bo \eta \int_{S^2} \!\!\!\!\! d\hat{\bo u} \; L_{\un,\um} (\bo \xi,  \bo \eta, s,\sigma;\bo x) \; e^{ \f{i}{\ve} \Theta_{\un,\um}^{k} (\bo \xi, \bo \eta,  s,\sigma, \hat{\bo u};\bo x)}
\label{Jknm}\\
&&\mathcal J^{\ve,k,l}_{\un,\um}(t,\bo x) =  - \f{ \mathcal N_{\ve}}{\ve^{5/2}}  \int_0^t \!\! ds \! \!\int_0^s \!\! d\sigma\;\!\! \int\!\!d\bo \xi \!\!\int\!\!d\bo \eta \int_{S^2} \!\!\!\!\! d\hat{\bo u} \; G_{\un,\um} (\bo \xi,\bo \eta, s,\sigma;\bo x) \; e^{ \f{i}{\ve}\Theta_{\un,\um}^{k,l} (\bo \xi, \bo \eta,  s,\sigma, \hat{\bo u};\bo x)}
\label{Jklnm}\\
&&L_{\un,\um} (\bo \xi,\bo \eta, s,\sigma;\bo x) = e^{i (\bo \eta+\bo \xi) \cdot \bo x + i(\f{s}{2} \bo \eta^2+\f{\sigma}{2}\bo \xi^2)+is\bo\eta\cdot\bo\xi}g_{\un,\um}(\bo \eta)g_{\um,\uo}(\bo \xi)\, f(\bo x+\sigma\,\bo \xi +s\,\bo\eta)
\label{Lnm}\\
&&G_{\un,\um} (\bo \xi,\bo \eta, s,\sigma;\bo x) = e^{i (\bo \eta+\bo \xi) \cdot \bo x + i(\f{s}{2} \bo \eta^2+\f{\sigma}{2}\bo \xi^2)+is\bo\eta\cdot\bo\xi}g_{\un,\uo}(\bo \eta)g_{\um,\uo}(\bo \xi)\, f(\bo x+\sigma\,\bo \xi +s\,\bo\eta)
\label{Gnm}\\
&&\Theta_{\un,\um}^{k} (\bo \xi,\bo \eta,  s,\sigma, \hat{\bo u};\bo x)\!\!= - (\bo \xi \!+\!\bo \eta)\! \cdot \bo a_k + v_0\, \hat{\bo u} \!\cdot (\bo x \!+\! \sigma\,\bo \xi \!+\! s\,\bo\eta)\! +(|n| \!-\! |m|)s+|m|\sigma
\label{thknm}\\
&&\Theta_{\un,\um}^{k,l} (\bo \xi, \bo \eta, s,\sigma, \hat{\bo u};\bo x)= - \bo \xi \cdot \bo a_l - \bo \eta \cdot \bo a_k + v_0\, \hat{\bo u} \cdot (\bo x \!+\! \sigma\,\bo \xi \!+\! s\,\bo\eta) +|n| s+|m|\sigma
\label{thklnm}
\ea
\end{prop}

\n
The proof proceeds along the same line of the previous one and it is omitted for the sake of brevity.


\section{Derivation of the leading term}

\n
In this section we analyze the term  $I^{\ve}_j (t) \Psi^{\ve}_{0,j}$. In particular,  we shall prove  decomposition (\ref{deco}), with an explicit expression for $Q_j^{\ve} (t)$.
We start from formulas (\ref{Ij}), (\ref{cIj}) proved in  proposition \ref{prop1}. In particular, by the change of variables $\bs{\xi} \rightarrow \mathcal R_j \bs{\xi}$, $\hat{\bs{u}} \rightarrow \mathcal R_j \hat{\bs{u}}$, with $\mathcal R_j$ defined in (\ref{mr}), we obtain
\ba\label{cIj0}
&&\mathcal I^{\ve}_{j, \un}(t,\bs{x}) =  - \f{i\, \mathcal N_{\ve}}{\ve^{5/2}}  \int_0^t \!\!\!ds \!\! \int\!\!d \bs{\xi} \!\! \int_{\mathcal C_0} \!\!\! d\hat{\bs{u}} \; F_{\un,j}^0 ( \bs{\xi},s;\bs{x}) \; e^{ \f{i}{\ve} \Phi_{\un,j}^0 (\bs{\xi}, s,  \hat{\bs{u}};\bs{x})}
\ea
where
\ba
&&F_{\un,j}^0 (\bs{\xi},s;\bs{x}) \equiv F_{\un} (\mathcal R_j^{-1}\bs{\xi},s;\bs{x}) = e^{i \bs{\xi} \cdot \bs{x}^j +i \f{s}{2} \bs{\xi}^2} g_{\un,0}(\mathcal R_j^{-1}\bs{\xi})\, f(\bs{x}^j+s \,\bs{\xi})\\
&&\Phi_{\un,j}^0 (\bs{\xi}, s, \hat{\bs{u}};\bs{x}) \equiv  \Phi_{\un,j} (\mathcal R_j^{-1}\bs{\xi}, s, \mathcal R_j^{-1} \hat{\bs{u}};\bs{x})= - |\bs{a}_j | \, \xi_3  + v_0\, \hat{\bs{u}} \cdot (\bs{x}^j+ s\, \bs{\xi}) +|n| s\\
&&\bs{x}^j \equiv \mathcal R_j \bs{x}
\ea
Notice that
\ba\label{xj3}
&&x^j_3= (0,0,1)\cdot \mathcal R_j \bs{x} = \hat{\bs{a}}_j \cdot \bs{x} \\
&&(x^j_1, x^j_2, 0) = \bs{x}^j - (\hat{\bs{a}}_j \cdot \bs{x} ) (0,0,1) = \mathcal R_j \left( \bs{x} - (\hat{\bs{a}}_j \cdot \bs{x} ) \hat{\bs{a}}_j \right)
\label{xj12}
\ea
Moreover we parametrize the unit vector $\hat{\bs{u}} \in \mathcal C_0$ as follows
\be
\hat{\bs{u}} = \left( \mu, \nu, \sqrt{1-\mu^2  -\nu^2} \right), \;\;\;\; (\mu, \nu) \in D_0 \equiv \left\{ (a,b) \in \erre^2, \; a^2 + b^2 < \sin^2 \theta_0 \right\}
\ee
Therefore the integral (\ref{cIj0}) is rewritten as
\ba\label{cIj2}
&&\mathcal I^{\ve}_{j, \un}(t,\bs{x}) =  - \f{i\, \mathcal N_{\ve}}{\ve^{5/2}}  \int\!\!d \bs{\xi} \!\!   \int_0^t \!\!\!ds  \!\! \int_{D_0} \!\!\! d\mu d\nu \; F_{\un,j}^1 ( \bs{\xi}, \mu,\nu,s ;\bs{x}) \; e^{ \f{i}{\ve} \Phi_{\un,j}^1 (\bs{\xi},  \mu, \nu,s ;\bs{x})}
\ea
where
\ba\label{F^1}
&&F_{\un,j}^1 ( \bs{\xi}, \mu,\nu,s;\bs{x}) \equiv \f{1}{\sqrt{1\!-\!\mu^2\!-\!\nu^2}} F_{\un,j}^0 ( \bs{\xi},s;\bs{x}) = \f{1}{\sqrt{1\!-\!\mu^2\!-\!\nu^2}}   e^{i \bs{\xi} \cdot \bs{x}^j +i \f{s}{2} \bs{\xi}^2} g_{\un,0}(\mathcal R_j^{-1}\bs{\xi})\, f(\bs{x}^j+s \,\bs{\xi})\nonumber\\
&&\\
&& \Phi_{\un,j}^1 (\bs{\xi},  \mu, \nu,s ;\bs{x})\!= -|\bs{a}_j |\, \xi_3 +v_0 \! \left[
\mu (\bs{x}^j_1 + s \xi_1) \!+ \nu (\bs{x}^j_2 + s \xi_2) \!+\! \sqrt{1\!-\! \mu^2 \!-\! \nu^2} (\bs{x}^j_3 +s \xi_3) \right] \!+\! |n| s\nonumber\\
&&
\ea
By a straightforward computation one can verify that the phase $ \Phi_{\un,j}^1 $ has exactly one (non degenerate) critical point in the integration region given by $\bs{\xi}= \bs{\xi}^c$, $(\mu,\nu,s)=\bs{z}^c$, where
\be\label{crit}
\bs{\xi}^c = \left( -\tau_j^{-1} x^j_1, -\tau_j^{-1} x^j_2, - v_0^{-1} |n|) \,, \;\;\;\;\;\;\; \bs{z}^c= (0, 0, \tau_j\right)
\ee
Therefore by stationary phase methods we can compute the asymptotic expansion of the oscillatory integral for $\ve \rightarrow 0$. In the specific case, we can exploit the  linearity of the phase in the variable $\bs{\xi}$ to obtain the expansion  in a relatively direct way. In particular we write
\ba
&& \Phi_{\un,j}^1 (\bs{\xi},  \mu, \nu,s ;\bs{x})= \bs{A}_j(\mu,\nu,s)\cdot \bs{\xi} +B_{\un,j}(\mu,\nu,s;\bs{x})\\
&&\bs{A}_j(\mu,\nu,s)= v_0 \left(\mu s, \nu s,  \sqrt{1\!-\! \mu^2 \!-\! \nu^2} \, s -\tau_j \right)\\
&&B_{\un,j}(\mu,\nu,s;\bs{x})= v_0 x^j_1 \mu + v_0 x^j_2 \nu + v_0 x^j_3   \sqrt{1\!-\! \mu^2 \!-\! \nu^2}  +|n|s
\ea
and
\ba\label{cIj4}
&&\mathcal I^{\ve}_{j, \un}(t,\bs{x}) =  - \f{i\, \mathcal N_{\ve}}{\ve^{5/2}}    \int_0^t \!\!\!ds  \!\! \int_{D_0} \!\!\! d\mu d\nu \; e^{\f{i}{\ve} B_{\un,j}(\mu,\nu,s;\bs{x})} \int\!\!d\bs{\xi}\;
 F_{\un,j}^1 ( \bs{\xi}, \mu,\nu,s ;\bs{x}) \; e^{ \f{i}{\ve} \bs{A}_j ( \mu, \nu,s) \cdot \bs{\xi}}
\ea
Let us introduce the following linear change of coordinates
\ba\label{chco}
&&(\mu,\nu,s) = L_{\ve} (z_1,z_2,z_3) \equiv L_{\ve} \bs{z}\\
&&\mu= \f{\ve}{v_0 \tau_j} z_1, \;\;\;\; \nu= \f{\ve}{v_0 \tau_j} z_2, \;\;\;\; s= \tau_j +\f{\ve}{v_0} z_3
\ea
The domain of integration in the variable $\bs{z}$ is
\ba
&&\Omega_{\ve} = \left\{ \bs{z} \in \erre^3\,|\, z_1^2+z_2^2 < \ve^{-2}v_0^2\tau_j^2 \sin^2 \theta_0,\; -\ve^{-1} v_0 \tau_j < z_3 < \ve^{-1}v_0 (t-\tau_j)  \right\}
\ea
We notice that for $\ve \rightarrow 0$ one has $L_{\ve} \bs{z} \rightarrow \bs{z}^c$, and $\Omega_{\ve} \rightarrow \erre^3$.
In the new integration variables $\bs{z}$ the integral (\ref{cIj4}) reads
\ba\label{cIj5}
&&\mathcal I^{\ve}_{j, \un}(t,\bs{x}) = -\f{i \mathcal N_{\ve} \sqrt{\ve}}{v_0^3 \tau_j^2} \int_{\Omega_{\ve}} \!\!\!\! d\bs{z}\; e^{\f{i}{\ve} B_{\un,j}( L_{\ve}\bs{z};\bs{x})}\!\!  \int \!\! d\bs{\xi} \; F^1_{\un,j}(\bs{\xi}, L_{\ve} \bs{z};\bs{x}) \; e^{\f{i}{\ve} \bs{A}_j ( L_{\ve}\bs{z})\cdot \bs{\xi}}
\ea
Let us expand $\bs{A}_j ( L_{\ve}\bs{z})$ and $B_{\un,j}( L_{\ve}\bs{z};\bs{x})$ around $\ve=0$
\ba
&&\bs{A}_j( L_{\ve}\bs{z}) =\! \left(\! \ve z_1 + \f{\ve^2}{v_0 \tau_j} z_1 z_3,  \ve z_2 + \f{\ve^2}{v_0 \tau_j} z_2 z_3,  \ve z_3 +\! \left(\sqrt{ 1- \f{\ve^2}{v_0^2 \tau_j^2}(z_1^2 + z_2^2)}  -1 \! \right) \!\left(  v_0\tau_j +\ve z_3 \right) \!\right) \nonumber\\
&&=\ve \bs{z} + \ve^2 \left( \f{z_1 z_3}{v_0 \tau_j}, \f{z_2 z_3}{v_0 \tau_j}, \ve^{-2} \! \left(  \sqrt{ 1- \f{\ve^2}{v_0^2 \tau_j^2}(z_1^2 + z_2^2)}  -1 \! \right) \!\left(  v_0\tau_j +\ve z_3 \right) \!\right) \\
&&\equiv \ve \bs{z} + \ve^2 \bs{A}_{j,\ve}^2 (\bs{z})
\ea
\ba
&&B_{\un,j}(L_{\ve}\bs{z};\bs{x})= v_0 x^j_3 +|n|\tau_j + \ve \f{x^j_1}{\tau_j} z_1 +\ve \f{x^j_2}{\tau_j} z_2 +\ve \f{|n|}{v_0} z_3 + v_0 x^j_3  \! \left(  \sqrt{ 1- \f{\ve^2}{v_0^2 \tau_j^2}(z_1^2 + z_2^2)}  -1 \! \right) \nonumber\\
&&\\
&&\equiv v_0 \, \hat{\bs{a}}_j \cdot \bs{x}  +|n|\tau_j  - \ve \, \bs{\xi}^c\! \cdot \bs{z} + \ve^2 B^2_{j,\ve}(\bs{z};\bs{x})
\ea
where we have used $x^j_3= (0,0,1)\cdot \mathcal R_j \bs{x}=\hat{\bs{a}}_j \cdot \bs{x}$ and (\ref{crit}).  Correspondingly, the integral (\ref{cIj5}) can be written as
\ba\label{cIj6}
&&\mathcal I^{\ve}_{j, \un}(t,\bs{x}) = -\f{i \mathcal N_{\ve} \sqrt{\ve}}{v_0^3 \tau_j^2} \,
e^{\f{i}{\ve} ( v_0 \, \hat{\bs{a}}_j \cdot \bs{x}  +|n|\tau_j ) }  \!\!\!
\int_{\Omega_{\ve}}\! \!\!\!\! d\bs{z}\, e^{- i \bs{\xi}^c \cdot \bs{z} +i \ve B^2_{j,\ve} (\bs{z};\bs{x} )   }\!\!  \int\! \!\! d\bs{\xi} \, F^1_{\un,j}(\bs{\xi}, L_{\ve} \bs{z};\bs{x}) \; e^{i \bs{z} \cdot \bs{\xi} +i \ve  \bs{A}^2_{j,\ve} ( \bs{z})}\nonumber\\
&&\equiv -\f{i \mathcal N_{\ve} \sqrt{\ve}}{v_0^3 \tau_j^2} \, e^{\f{i}{\ve} \left(  v_0 \, \hat{\bs{a}}_j \cdot \bs{x}  +|n|\tau_j  \right)}
 (2\pi)^3 F^1_{\un,j} (\bs{\xi}^c, \bs{z}^c;\bs{x}) +Q^{\ve}_{j,\un} (t,\bs{x})
 \ea
 where
 \ba\label{Q}
 &&Q^{\ve}_{j,\un} (t,\bs{x}) = -\f{i \mathcal N_{\ve} \sqrt{\ve}}{v_0^3 \tau_j^2} \, e^{\f{i}{\ve} \left(  v_0 \, \hat{\bs{a}}_j \cdot \bs{x}  +|n|\tau_j  \right)} \!
 \left[   \int_{\erre^3 \setminus \Omega_{\ve}}\!\! \!\!\!\!d \bs{z} \, e^{-i \bs{\xi}^c \cdot \bs{z}} \!\!\int\!\! d\bs{\xi} \, e^{i \bs{z} \cdot \bs{\xi}}    F^1_{\un,j}(\bs{\xi}, \bs{z}^c ;\bs{x}) \right. \nonumber\\
&& \left. + \! \int_{\Omega_{\ve}} \!\!\!\!d \bs{z} \, e^{-i \bs{\xi}^c \cdot \bs{z}} \!\!\!\int\!\! d\xi \, e^{i \bs{z} \cdot \bs{\xi}} \! \left(\!   F^1_{\un,j}(\bs{\xi}, L_{\ve} \bs{z};\bs{x}) e^{i\ve \left( B^2_{j,\ve}(\bs{z};\bs{x}) + \bs{A}^2_{j,\ve} (\bs{z} )  \right)}    -        F^1_{\un,j}(\bs{\xi}, \bs{z}^c;\bs{x})  \!      \right)\right]
\ea
Let us compute $F^1_{\un,j}   (\bs{\xi}^c, \bs{z}^c;\bs{x})   $ (see (\ref{F^1})). Taking into account that $\bs{z}^c =(0,0,\tau_j)$ and
\ba\label{xic}
&&\!\!\!\bs{\xi}^c \!= -\tau_j^{-1} (x^j_1,x^j_2,0) - |n| v_0^{-1} (0,0,1) \!= -\tau_j^{-1} \mathcal R_j \big( \bs{x}- (\hat{\bs{a}}_j \cdot \bs{x} ) \hat{\bs{a}}_j \big) - |n| v_0^{-1} \mathcal R_j \hat{\bs{a}}_j
\ea
we have
\ba\label{F^11}
&&F^1_{\un,j} (\bs{\xi}^c, \bs{z}^c;\bs{x}) = e^{i \bs{\xi}^c \cdot \bs{x}^j + \f{i}{2} \tau_j (\bs{\xi}^c)^2 } g_{\un,j}(\mathcal R_j^{-1} \bs{\xi}^c ) \, f( \bs{x}^j +\tau_j \bs{\xi}^c)\nonumber\\
&&= e^{i \left[ -\f{1}{\tau_j} \big( \bs{x} - ( \hat{\bs{a}}_j \cdot \bs{x} ) \hat{\bs{a}}_j \big) \right] \cdot \bs{x} - i \f{|n|}{v_0} \hat{\bs{a}}_j \cdot \bs{x}  + i \f{\tau_j}{2} \left( \f{1}{\tau_j^2} |\bs{x} - ( \hat{\bs{a}}_j \cdot \bs{x} ) \hat{\bs{a}}_j |^2 +\f{|n|^2}{v_0^2} \right)   }
g_{\un,j} \left(  -\tau_j^{-1}  \big( \bs{x}- (\hat{\bs{a}}_j \!\cdot \!  \bs{x} ) \hat{\bs{a}}_j \big) - |n| v_0^{-1} \hat{\bs{a}}_j
\right)\nonumber\\
&&\cdot \,  f\left( \hat{\bs{a}}_j \!  \cdot \! \bs{x} - |n| v_0^{-1} \tau_j \right)\nonumber\\
&&= e^{i \f{|n|^2 \tau_j}{2 v_0^2}   -i \f{|n|}{v_0} \hat{\bs{a}}_j  \cdot \bs{x} - i \f{| \bs{x}- (\hat{\bs{a}}_j \!\cdot \!  \bs{x} ) \hat{\bs{a}}_j   |^2}{2\tau_j} }
g_{\un,j} \! \left(  -\tau_j^{-1}  \big( \bs{x}- (\hat{\bs{a}}_j \!\cdot \!  \bs{x} ) \hat{\bs{a}}_j \big) \! - |n| v_0^{-1} \hat{\bs{a}}_j
\right)  f\left( \hat{\bs{a}}_j \!  \cdot \! \bs{x} - |n| v_0^{-1} \tau_j \right)
\ea
From (\ref{cIj6}), (\ref{Q}), (\ref{F^11}) and definition \ref{depac} we obtain
\ba\label{cIj7}
&&\mathcal I_{\un,j}^{\ve} (t,\bs{x}) = \ve^2 P_{\un,j}^{\ve} (\ve \bs{x}) + Q_{\un,j}^{\ve} (t, \bs{x})
\ea
Taking into account of (\ref{cIj7}) and (\ref{Ij}), we have proved the decomposition (\ref{deco}), with $Q_j^{\ve}(t)$ explicitly given by
\be
Q_j^{\ve}(t,\bs{x}, \bs{r}_1,\ldots,\bs{r}_N)= \sum_{\un} Q_{\un,j}^{\ve}(t,\bs{x})  \;   \varphi^{\ve}_{n,j} (\bs{r}_j) \prod_{k, k \neq j}  \varphi^{\ve}_{0 ,k} (\bs{r}_k)
\ee


\section{Estimate of $Q_j^{\ve}(t) $}

Let us first recall some elementary facts that will be used in this and in the following sections.  The unitary propagator $U(t)$ of the harmonic oscillator centered in the origin and
corresponding to $\hbar=m=\omega=1$ has an  integral kernel  given by
\begin{equation}\label{nucleoint}
U(t,\bo x,\bo y)=\sum_{\un} \phi_{\un}(\bo x)\phi_{\un}(\bo y)e^{-it\left(|n|+\frac{3}{2}\right)}
\end{equation}
Using such representation we can compute the following sum
\begin{eqnarray}\label{sumn1t0}
\sum_{\un}e^{- i
\frac{|n|}{\ve} t } g_{\un,\uo}(\bo \xi)\bar{g}_{\un,\uo}(\bo \xi')=\nonumber\\
\widetilde{V}(\bo \xi)\overline{\widetilde{V}}(\bo \xi')
\frac{1}{(2\pi)^3}\sum_{\un}\int d\bo y \phi_{\un}(\bo y)\phi_{\uo}(\bo y)e^{-i\bo y\cdot \bo \xi}\int d\bo y' \phi_{\un}(\bo y')\phi_{\uo}(\bo y')e^{i\bo y'\cdot \bo \xi'}e^{- i
\frac{|n|}{\ve} t }=\nonumber\\
\widetilde{V}(\bo \xi)\overline{\widetilde{V}}(\bo \xi')
\frac{1}{(2\pi)^3}\int d\bo y'\,\overline{\phi_{\uo}(\bo y')e^{-i\bo y'\cdot \bo \xi'}}\int d\bo y \sum_{\un}\phi_{\un}(\bo y)\phi_{\un}(\bo y')e^{-i
\frac{|n|}{\ve} t }\phi_{\uo}(\bo y)e^{-i\bo y\cdot\bo \xi}=\nonumber\\
\widetilde{V}(\bo\xi)\overline{\widetilde{V}}(\bo \xi')
\frac{e^{i\frac{3}{2\varepsilon} t }}{(2\pi)^3}
\langle\phi_{\uo},  e^{i \bo \xi' \cdot (\cdot) } U(t/\ve)  e^{- i \bo \xi \cdot (\cdot) } \phi_{\uo} \rangle
\end{eqnarray}
In particular for $t=0$ one has
\begin{eqnarray}\label{sm}
\sum_{\un} g_{\un,\uo}(\bo \xi)\bar{g}_{\un,\uo}(\bo \xi')=\frac{\widetilde{V}(\bo\xi)\overline{\widetilde{V}}(\bo \xi')}{(2\pi)^3}\, e^{-\frac{(\bo \xi-\bo\xi')^2}{4}}
\end{eqnarray}

\n
Let us now  prove the estimate (\ref{stiQ}).
Taking into account  of $(5.29)$ and (5.25) we write
\ba\label{q}
&& Q^{\ve}_j (t)(t, \bo x,  \bo r_1,\ldots, \bo r_N)
=   \sum_{\un}  \left( Q^{\ve,1}_{j, \un}(t,\bo x)+ Q^{\ve,2}_{j, \un}(t,\bo x)\right)  \varphi^{\ve}_{\un,j}(\bo r_j)  \prod_{k,\,k\neq j}\varphi^{\ve}_{\uo,k} (\bo r_k),
\ea
where
\ba
\label{q1} Q^{\ve,1}_{j, \un}(t,\bo x) =&&\!\!\!\!\!\!
 -\f{i \mathcal N_{\ve} \sqrt{\ve}}{v_0^3 \tau_j^2} \, e^{\f{i}{\ve} \left(  v_0 \, \hat{\bs{a}}_j \cdot \bs{x}  +|n|\tau_j  \right)}  \int_{\erre^3 \setminus \Omega_{\ve}}\!\! \!\!\!\!d \bs{z} \, e^{-i \bs{\xi}^c \cdot \bs{z}} \!\!\int\!\! d\bs{\xi} \, e^{i \bs{z} \cdot \bs{\xi}}    F^1_{\un,j}(\bs{\xi}, \bs{z}^c ;\bs{x})  \\
Q^{\ve,2}_{j, \un}(t,\bo x)=&&\!\!\!\!\!\!
-\f{i \mathcal N_{\ve} \sqrt{\ve}}{v_0^3 \tau_j^2} \, e^{\f{i}{\ve} \left(  v_0 \, \hat{\bs{a}}_j \cdot \bs{x}  +|n|\tau_j  \right)} \! \! \int_{\Omega_{\ve}} \!\!\!\!d \bs{z} \, e^{-i \bs{\xi}^c \cdot \bs{z}} \!\!\!\int\!\! d\bo \xi \, e^{i \bs{z} \cdot \bs{\xi}} \nonumber\\\label{q2}
&&\!\!\!\!\!\cdot \left(\!   F^1_{\un,j}(\bs{\xi}, L_{\ve} \bs{z};\bs{x}) e^{i\ve \left( B^2_{j,\ve}(\bs{z};\bs{x}) + \bs{A}^2_{j,\ve} (\bs{z} )  \right)}    -        F^1_{\un,j}(\bs{\xi}, \bs{z}^c;\bs{x})  \!      \right) \ea
and, from $(5.9)$,$(5.11)$,$(5.16)$,$(5.17)$, the explicit expressions for $F^1_{\un,j}(\bs{\xi}, \bs{z}^c ;\bs{x}) $ and  $F^1_{\un,j}(\bs{\xi}, L_{\ve} \bs{z};\bs{x}) $ are given by
\ba \label{f1} F^1_{\un,j}(\bs{\xi}, L_{\ve} \bs{z};\bs{x})=\f{1}{\sqrt{1\!-\!\ \frac{\ve^2}{v_0^2\tau_j^2}(z_1^2\!+\!z_2^2)}}   e^{i \bs{\xi} \cdot \bs{x}^j +i \f{\tau_j}{2} \bs{\xi}^2+
i \ve\f{z_3}{2v_0} \bs{\xi}^2} g_{\un,0}(\mathcal R_j^{-1}\bs{\xi})\, f\left(\bs{x}^j+\tau_j \,\bs{\xi} + \ve\f{z_3}{v_0} \bs{\xi}\right)\ea
\ba \label{f1c} F^1_{\un,j}(\bs{\xi}, \bs{z}^c;\bs{x})=  e^{i \bs{\xi} \cdot \bs{x}^j +i \f{\tau_j}{2} \bs{\xi}^2} g_{\un,0}(\mathcal R_j^{-1}\bs{\xi})\, f\left(\bs{x}^j+\tau_j \,\bs{\xi} \right)\ea
 From (\ref{q}) we have
\be\label{stiq1}\|Q_j^{\ve} (t)\|^2\leq 2\ve^3\sum_{\un}\int d\bo x\,\left| Q^{\ve,1}_{j, \un}(t,\bo x)\right|^2+2\ve^3\sum_{\un}\int d\bo x\,\left| Q^{\ve,2}_{j, \un}(t,\bo x)\right|^2
\ee
Let us  estimate the first term in the r.h.s of (\ref{stiq1}). Using  (\ref{q1}) and (\ref{f1c}) we have
\ba\label{stiq2}
&&\sum_{\un}\int \!\! d\bo x\,\left| Q^{\ve,1}_{j, \un}(t,\bo x)\right|^2 \nonumber\\
&&=
\f{\mathcal N_{\ve}^2 \ve}{v_0^6 \tau_j^4}\sum_{\un}
\int \!\! d\bo x
\int_{\erre^3 \setminus \Omega_{\ve}}\!\! \!\!\!\!d \bs{z} \,
e^{-i \bs{\xi}^c \cdot \bs{z}} \!\!
\int\!\! d\bs{\xi} \, e^{i \bs{z} \cdot \bs{\xi}}    e^{i \bs{\xi} \cdot \bs{x}^j +i \f{\tau_j}{2} \bs{\xi}^2} g_{\un,0}(\mathcal R_j^{-1}\bs{\xi})\, f\left(\bs{x}^j+\tau_j \,\bs{\xi} \right) \nonumber\\
&&\cdot \int_{\erre^3 \setminus \Omega_{\ve}}\!\! \!\!\!\!d \bs{z}' \,
e^{i \bs{\xi}^c \cdot \bs{z}'}
\!\!\int\!\! d\bs{\xi}' \, e^{-i \bs{z}' \cdot \bs{\xi}'}
e^{-i \bs{\xi}' \cdot \bs{x}^j -i \f{\tau_j}{2} \bs{\xi}^{'2}} \bar{g}_{\un,0}(\mathcal R_j^{-1}\bs{\xi}')\, f\left(\bs{x}^j+\tau_j \,\bs{\xi}' \right)
\nonumber\\
&&= \f{\mathcal N_{\ve}^2 \ve}{v_0^6 \tau_j^4 (2\pi)^3} \!
\int_{\erre^3 \setminus \Omega_{\ve}}\!\! \!\!\!\!d \bs{z}
 \int_{\erre^3 \setminus \Omega_{\ve}}\!\! \!\!\!\!d \bs{z}'\!
\int \!\! d\bo x \, e^{-i \bs{\xi}^c \cdot (\bs{z} - \bs{z}')} \!\!
\int\!\! d\bs{\xi}   \! \int\!\! d\bs{\xi}'
\, e^{i \bs{z} \cdot \bs{\xi} - i \bs{z}' \cdot \bs{\xi}' }
e^{i (\bs{\xi} - \bo \xi')  \cdot \bs{x}^j +i \f{\tau_j}{2} (\bs{\xi}^2 - \bo \xi'^2)  } \nonumber\\
&&\cdot  \, \widetilde{V}(\mathcal R_j^{-1}\bs{\xi})\,\overline{\widetilde{V}}(\mathcal R_j^{-1}\bs{\xi}')\,
e^{- \f{ (\bo \xi - \bo \xi')^2}{4}}
f\left(\bs{x}^j+\tau_j \,\bs{\xi} \right)  f\left(\bs{x}^j+\tau_j \,\bs{\xi}' \right)
\ea
where the sum over $\un$ in (\ref{stiq2}) has been explicitly computed exploiting (\ref{sm}).  Using the identity
\be
e^{i \bo z \cdot \bo \xi} = -\f{1}{|\bo z|^2} \Delta_{\bo \xi} \, e^{i \bo z \cdot \bo \xi} = \f{1}{|\bo z|^4} (\Delta_{\bo \xi})^2 e^{i \bo z \cdot \bo \xi}
\ee
and integrating by parts we find
\ba
&&\sum_{\un}  \!  \int \!\!  d\bo x \left| Q^{\ve,1}_{j, \un}(t,\bo x)\right|^2
\nonumber\\
&&= \f{\mathcal N_{\ve}^2 \ve}{v_0^6 \tau_j^4 (2\pi)^3} \!
\int_{\erre^3 \setminus \Omega_{\ve}}\!\! \!\!\!\!d \bs{z} \, \f{1}{|\bo z|^4}
 \int_{\erre^3 \setminus \Omega_{\ve}}\!\! \!\!\!\!d \bs{z}'    \, \f{1}{|\bo z'|^4}
 \!  \int \!\! d\bo x \, e^{-i \bs{\xi}^c \cdot (\bs{z} - \bs{z}')} \!\!
  \int\!\! d\bs{\xi}   \! \int\!\! d\bs{\xi}'
\, e^{i \bs{z} \cdot \bs{\xi} - i \bs{z}' \cdot \bs{\xi}' } (\Delta_{\bo \xi})^2  (\Delta_{\bo \xi'})^2  \nonumber\\
&&
\left( \!
\widetilde{V}(\mathcal R_j^{-1}\bs{\xi})\,\overline{\widetilde{V}}(\mathcal R_j^{-1}\bs{\xi}')\,
e^{- \f{ (\bo \xi - \bo \xi')^2}{4}}
e^{i (\bs{\xi}   \cdot \bs{x}^j + \f{\tau_j}{2} \bs{\xi}^2 )  }
f\left(\bs{x}^j+\tau_j \,\bs{\xi} \right)
e^{- i (  \bo \xi'   \cdot \bs{x}^j +  \f{\tau_j}{2}  \bo \xi'^2)  }
 f\left(\bs{x}^j+\tau_j \,\bs{\xi}' \right)
\! \right)  \nonumber\\
&&\leq c \, \ve \! \int_{\erre^3 \setminus \Omega_{\ve}}\!\! \!\!\!\!d \bs{z} \, \f{1}{|\bo z|^4} \!
 \int_{\erre^3 \setminus \Omega_{\ve}}\!\! \!\!\!\!d \bs{z}'    \, \f{1}{|\bo z'|^4}  \!
 \int \!\! d\bo \xi \, \langle \bo \xi \rangle^4 \! \!\! \sum_{\underline{m}, |m|\leq4} |D^{\underline{m}}_{\bo \xi} \widetilde{V}(\mathcal R_j^{-1}\bs{\xi}) |
  \int \!\! d\bo \xi' \, \langle \bo \xi' \rangle^4 \!\! \! \sum_{\underline{n}, |n|\leq4} |D^{\underline{n}}_{\bo \xi} \widetilde{V}(\mathcal R_j^{-1}\bs{\xi}') | \nonumber\\
&&\cdot \sum_{\underline{k}, |k| \leq4}  \sum_{\underline{l}, |l| \leq4} \int\!\! d \bo x\,  \langle \bo x^j +\tau_j \bo \xi \rangle^4
|D^{\underline{k}}_{\bo \xi} f(\bo x^j + \tau_j \bo \xi) |
 \langle \bo x^j +\tau_j \bo \xi' \rangle^4
|D^{\underline{l}}_{\bo \xi'} f(\bo x^j + \tau_j \bo \xi') |
\ea
Furthermore, in the last integral we apply  the Schwartz inequality and we obtain
\ba\label{q1fin}
&&\sum_{\un}  \!  \int \!\!  d\bo x \left| Q^{\ve,1}_{j, \un}(t,\bo x)\right|^2 \leq c \,  \| \widetilde{V}\|^2_{W^{4,1}_4} \left(\!  \int_{\erre^3 \setminus \Omega_{\ve}}\!\! \!\!\!\!d \bs{z} \, \f{1}{|\bo z|^4} \!\right)^{\!2} \ve\,  \leq c \,   \| \widetilde{V}\|^2_{W^{4,1}_4}\,   \ve^3
\ea
where the last integral in (\ref{q1fin}) has been explicitly computed.

\n
The next step is to estimate the second term in the r.h.s of (\ref{stiq1}).
We consider the Taylor expansion of the $\ve$-dependent part of the integrand in (\ref{q2}) up to order one. We have
\ba\label{st}
F^1_{\un,j}(\bs{\xi}, L_{\ve} \bs{z};\bs{x}) e^{i\ve \left( B^2_{j,\ve}(\bs{z};\bs{x}) + \bs{A}^2_{j,\ve} (\bs{z} )  \right)}    -        F^1_{\un,j}(\bs{\xi}, \bs{z}^c;\bs{x})= g_{\un,\uo} (\mathcal R^{-1}_j \bo \xi) \,
\ve \!   \! \int_0^1 \!\!\! d\vartheta \,  \frac{\partial T_j}{\partial \ve}(\vartheta\ve,\bs{\xi}, \bs{z};\bs{x})\nonumber\\
\ea
where
\ba T_j(\ve,\bs{\xi}, \bs{z};\bs{x})\equiv
\f{ e^{ i \ve \left( B^2_{j,\ve}(\bs{z};\bs{x}) + \bs{A}^2_{j,\ve} (\bs{z} )  \right)} }{ \sqrt{ 1 - \f{ \ve^2}{v_0^2 \tau_j^2} ( z_1^2 + z_2^2)} } \,
e^{i \left( \bo \xi \cdot \bo x^j + \f{\tau_j}{2} \bo \xi^2 + \ve \f{z_3}{2 v_0} \bo \xi^2 \right) } f(\bo x^j +\tau_j \bo \xi + \ve \f{z_3}{v_0} \bo \xi )
\ea
Substituting  (\ref{st}) into  the second term in the r.h.s. of   (\ref{stiq1}) and computing the sum over $\un$ we obtain
\ba\label{stiq4}
&&\sum_{\un}\int \!\! d\bo x \left| Q^{\ve,2}_{j, \un}(t,\bo x)\right|^2 \nonumber\\
&&= \f{\mathcal N_{\ve}^2 \,\ve^3}{v_0^6 \tau_j^4 (2\pi)^3} \!
\int_{ \Omega_{\ve}}\!\! \!\!\!d \bs{z}
 \int_{ \Omega_{\ve}}\!\! \!\!\!d \bs{z}'\!
\int \!\! d\bo x \, e^{-i \bs{\xi}^c \cdot (\bs{z} - \bs{z}')} \!\!
\int\!\! d\bs{\xi}   \! \int\!\! d\bs{\xi}'
\, e^{i \bs{z} \cdot \bs{\xi} - i \bs{z}' \cdot \bs{\xi}' }
  \, \widetilde{V}(\mathcal R_j^{-1}\bs{\xi})\,\overline{\widetilde{V}}(\mathcal R_j^{-1}\bs{\xi}')\,
e^{- \f{ (\bo \xi - \bo \xi')^2}{4}} \nonumber\\
&&\cdot \int_0^1 \!\!\! d\vartheta \,  \frac{\partial T_j}{\partial \ve}(\vartheta\ve,\bs{\xi}, \bs{z};\bs{x})
\int_0^1 \!\!\! d\vartheta' \, \overline{  \frac{\partial T_j}{\partial \ve} }(\vartheta'\ve,\bs{\xi}', \bs{z}';\bs{x})
\ea
We can now proceed with integration by parts along the same line of the previous case and we find the same kind of result
\ba\label{stiq2fin}
&&\sum_{\un}\int \!\! d\bo x \left| Q^{\ve,2}_{j, \un}(t,\bo x)\right|^2 \leq c\,  \|\widetilde{V}\|^2_{W^{4,1}_4} \, \ve^3
\ea
Using the estimates (\ref{q1fin}), (\ref{stiq2fin}) in (\ref{stiq1}) we obtain the estimate (\ref{stiQ}) and therefore we also conclude the proof of proposition \ref{prop2}.

\hfill $\Box$


\section{Non stationary terms}

For the estimate of the non stationary terms  it will be useful the following technical lemma. 
\begin{lemma}\label{lemmazeta}
For any $s_1, \ldots ,s_{l-1} >0$, $l>1$, let $\;\zeta_{s_1,\ldots ,s_{l-1}}: \mathbb{R}^{3l}   \rightarrow \mathbb{C} $  be defined by
\begin{equation}\label{zeta}
\zeta_{s_1,\ldots ,s_{l-1}}(\bo{\xi_1},\cdots \bo \xi_l)
\equiv \langle \phi_{\uo}, e^{i\bo \xi_l\cdot (\cdot)} U\left(s_{l-1}\right) e^{-i \bo \xi_{l-1} \cdot (\cdot)} U(s_{l-2})
\cdots e^{-i\bo\xi_{2}\cdot (\cdot)} U\left(s_1\right)\, e^{-i\bo\xi_{1}\cdot (\cdot)}\phi_{\uo} \rangle
\end{equation}
Then $\zeta_{s_1,\ldots ,s_{l-1}} \in C^{\infty}(\mathbb{R}^{3l})
$ and  for every
$\underline{\alpha}_1 , \ldots, \underline{\alpha}_l \in \enne^3 $  there exists  $c_{\underline{\alpha}_1, \ldots , \underline{\alpha}_l }$,
independent of $s_1, \ldots , s_{l-1}$, such that the following estimate holds
\begin{equation}
\| D^{\underline{\alpha}_1  }_{\bo{\xi_1}}  \cdots D^{\underline{\alpha}_l}_{\bo \xi_l}  \zeta_{s_1,\cdots ,s_{l-1}}(\bo{\xi_1},\ldots , \bo \xi_l)\|\leq
c_{\underline{\alpha}_1 \ldots \underline{\alpha}_l  }\left(\sum_{i=1}^l\langle \bo\xi_i \rangle\right)^{|\alpha_1|+ \cdots +|\alpha_l| }
\end{equation}
\end{lemma}

\n
The proof is a direct generalization of the proof of the one dimensional case given in lemma $3.1$ in \cite{fint} and it is omitted. Moreover by formula (\ref{sumn1t0}) we have
\begin{eqnarray}\label{sumn1t}
\sum_{\un}e^{- i
\frac{|n|}{\ve} t } g_{\un,\uo}(\bo \xi)\bar{g}_{\un,\uo}(\bo \xi)=
\frac{e^{i\frac{3}{2\varepsilon} t }}{(2\pi)^3}\,
\widetilde{V}(\bo\xi)\overline{\widetilde{V}}(\bo \xi')\,
\zeta_{ t /\varepsilon}(\bo \xi,\bo\xi')
\end{eqnarray}

\n
Using these facts, in the following we prove proposition \ref{prop3}. 

\n
From the explicit expression of $\mathcal T^{\ve}_{j, \un}(t,\bo x)$ given in proposition \ref{prop1} we have
\ba \label{nonstaz1}
&&\|  I_j^{\ve}(t) \left( \Psi_0^{\ve} -  \Psi_{0,j}^{\ve} \right) \|^2 \nonumber\\
&&= \ve^3
\sum_{\un} \int d\bo x\,  \left| \mathcal T^{\ve}_{j, \un}(t,\bo x) \right|^2 = \f{ \mathcal N_{\ve}^2}{\ve^{2}} \sum_{\un} \!\!  \int \!\! d\bo x \! \int_0^t \!\!\!ds \!\! \int\!\!d\bo \xi \!\! \int_{S^2\backslash \mathcal C_j} \!\!\!\!\! \!\! d\hat{\bo u} \; F_{\un} (\bo \xi,s;\bo x) \; e^{ \f{i}{\ve} \Phi_{\un,j} (\bo \xi, s, \hat{\bo u};\bo x)} \! \nonumber\\
 &&\cdot \int_0^t \!\!\!ds' \!\! \int\!\!d\bo \xi' \!\! \int_{S^2\backslash \mathcal C_j}\!\!  \!\!\!\!\! d\hat{\bo u}' \; \overline{F_{\un}} (\bo \xi',s';\bo x) \; e^{ - \f{ i}{\ve} \Phi_{\un,j} (\bo \xi', s', \hat{\bo u}';\bo x)}
\ea
Taking into account the expression of $F_{\un}$ and $\Phi_{\un,j}$ (see  (\ref{Fn}), (\ref{gnm}), (\ref{phinj})) and formula (\ref{sumn1t0}), we compute the sum over $\un$  in (\ref{nonstaz1})
\begin{eqnarray}\label{nonstaz2}
&&\|  I_j^{\ve}(t) \left( \Psi_0^{\ve} -  \Psi_{0,j}^{\ve} \right) \|^2 =
\frac{\mathcal{N}_{\ve}^2}{\ve^2(2\pi)^3}\int\!\! d\bo x \int_0^t \!\!ds\int_0^t ds'\,e^{i\frac{3}{2\ve}(s'-s)}
\int_{S^2\backslash \mathcal C_j} \!\!\!\!\! d\hat{\bo u} \int_{S^2\backslash \mathcal C_j} \!\!\!\!\! d\hat{\bo u}' \; e^{\frac{i}{\ve}v_0(\hat{\bo u}-\hat{ \bo u}')\cdot \bo x}\nonumber\\
&& \cdot \! \int  \!\!  d\bo\xi' \int \!\!  d\bo\xi\,e^{-\frac{i}{\ve}  \Phi_j (\bo \xi', s', \hat{\bo u}' ) }
e^{\frac{i}{\ve}\Phi_j(\bo\xi,s,\hat{\bo u})}G_{\varepsilon}(\bo\xi,\bo \xi',s,s';\bo x)
\end{eqnarray}
where
\begin{eqnarray}
\label{G}
&& G_{\varepsilon}(\bo \xi,\bo \xi',s,s';\bo x) \equiv
 \tilde{V}(\bo\xi) \overline{\tilde{V}}(\bo \xi') \,
 e^{i \left( \bo x\cdot
\bo \xi  +  \frac{s}{2} \bo \xi^2 \right) } f(\bo x+\!
s\bo \xi)\, 
 e^{-i \left( \bo x\cdot
\bo \xi' +  \frac{s'}{2} \bo \xi^{'2} \right) }    f(\bo x+\! s'\bo \xi')
\zeta_{(s-s')/\ve} (\bo \xi, \bo \xi')\nonumber\\
&&\\
&& \Phi_j(\bo \xi,s,\hat{\bo u}) \equiv  v_0 \,
\bo \xi\cdot\left(  s\hat{\bo u}-\tau_j\hat{\bo a}_j \right)
\end{eqnarray}
Let us consider the last two oscillatory integrals in the variables $\bo \xi, \bo \xi'$ in formula (\ref{nonstaz2}). We notice that 
\be
|\nabla_{\bo \xi} \Phi_j |^2 = v_0^2 (s^2 + \tau_j^2 - 2 s \tau_j\, \hat{\bo u} \cdot \hat{\bo a}_j ) = 
v_0^2 \left(s^2 + \tau_j^2 - 2 s \tau_j \, (\mathcal R_j \hat{\bo u})_3 \right)
\ee
For $\hat{\bo u} \in S^2 \setminus \mathcal C_j$ we have $ (\mathcal R_j \hat{\bo u} )_3 = \sqrt{ 1 -  (\mathcal R_j \hat{\bo u} )_1^2 -  (\mathcal R_j \hat{\bo u} )_2^2 } \leq \cos \theta_0$ and therefore 
\ba\label{cac}
&&|\nabla_{\bo \xi}\Phi_j |^2
\geq
v_0^2\left(s^2+\tau_j^2-2s\tau_j \cos \theta_0 \right)   \geq v_0^2  \min_s \left(s^2+\tau_j^2-2s\tau_j \cos \theta_0 \right) \nonumber\\
&&=
v_0^2 \tau_j^2 \sin^2 \theta_0 \geq v_0^2 \tau_1^2 \sin^2 \theta_0  \equiv \Delta^2
\ea
Using the above inequality, we can estimate the two integrals in the variables $\bo \xi , \bo \xi'$ in (\ref{nonstaz2})  exploiting a non-stationary phase argument. With a repeated application $k$ times  of   the identity
\begin{equation}\label{per_parti}
a\, e^{ib}=-i \, div \left(e^{ib}\frac{\nabla b}{|\nabla b|^2}a\right)+i\, e^{ib}div\left(\frac{\nabla b}{|\nabla b|^2}a\right)
\end{equation}
we have
\begin{eqnarray}\label{intparti}
&&\int \!\!d\bo \xi\int \!\!d\bo \xi'\,
e^{\frac{i}{\varepsilon}\Phi_j (\bo \xi,s,\hat{\bo u};\bo x)}
e^{-\frac{i}{\varepsilon}\Phi_j (\bo \xi',s',\hat{\bo u}';\bo x)}
G_{\varepsilon}(\bo \xi,\bo \xi',s,s';\bo x)
=\nonumber\\&&
\varepsilon^{2k}  \! \int \!\!
d\bo \xi\,\int \!\!d\bo \xi'\,
e^{\frac{i}{\ve}\Phi_j (\bo \xi,s,\hat{\bo u};\bo x)}
e^{-\frac{i}{\ve}\Phi_j (\bo \xi',s',\hat{\bo u}';\bo x)}
L_{\bo \xi}^{k}L_{\bo \xi'}^{k}
G_{\varepsilon}(\bo \xi,\bo \xi',s,s';\bo x)
\end{eqnarray}
where
\begin{equation}\label{L}
L_{\bo \xi}G_{\ve} =\sum_{j=1}^3 u_j \f{\partial G_{\ve}}{\partial \xi_j} , \;\;\;\;
\quad u_{j}
=\frac{1}{|\nabla_{\bo {\xi}}\Theta|^2}
\frac{\partial \Theta}{\partial \xi_{j}}
\end{equation}
and moreover
\begin{equation}\label{Lk}L_{\bo \xi}^kG_{\ve}=
\sum_{\underline{m}, |m|=k} u_1^{m_1}u_2^{m_2} u_3^{m_3}D^{\underline{m}}_{\bo \xi}G_{\ve}
\end{equation}
We notice that
\be
|u_{j} | \leq \f{1}{|\nabla_{\bo \xi} \Theta|} < \f{1}{\Delta}
\ee
Therefore
\ba \label{G1}
&& \left| \int \!\!  d\bo\xi\,\int \!\!  d\bo\xi'\,
G_{\varepsilon}(\bo\xi,\bo \xi',s,s';\bo x)
e^{\frac{i}{\ve}\Phi_j (\bo\xi,s,\hat{\bo u})} e^{-\frac{i}{\ve}\Phi_j (\bo\xi,s,\hat{\bo u})}   \right|   \nonumber\\
&&  \leq
\f{\ve^{2k}}{\Delta^{2k}}
\sum_{\underline{m},
|m|=k}
\; \sum_{ \underline{m}', 
|m'|=k}
  \int \!\! d\bo \xi \,  \int \!\! d\bo \xi' \, \left|  D^{\underline{m}'}_{\bo \xi'}D^{\underline{m}}_{\bo \xi} G_{\ve}(\bo\xi,\bo \xi',s,s';\bo x) \right|
\ea
The next step is to estimate the derivatives of  $G_{\varepsilon}$. From (\ref{G}) we have
\begin{eqnarray}\label{gtilde}
&&\sum_{\underline{m}, 
|m|=k} \; \sum_{\underline{m'}, 
|m'|=k} \left|  D^{\underline{m}'}_{\bo \xi'}D^{\underline{m}}_{\bo \xi} G_{\ve}(\bo\xi,\bo \xi',s,s';\bo x) \right| \nonumber\\
&& \leq  c_k \!\! \!
\sum_{\underline{n}, |n|\leq k }  \!\!
| D^{\underline{n}}_{\bo \xi} \zeta_{(s-s')/\varepsilon}(\bo \xi,\bo \xi')| \!\!
  \sum_{\underline{l}, |l|\leq k} \!\!  | D^{\underline{l}}_{\bo \xi} \widetilde{V}(\bo \xi)| 
    \sum_{\underline{l}', |l'|\leq k} \!\!  | D^{\underline{l}'}_{\bo \xi'} \widetilde{V}(\bo \xi')|\langle \bo x + s \bo \xi \rangle^k\langle \bo x + s '\bo \xi'\rangle^k  \nonumber\\
&&\langle s \rangle^k\langle s' \rangle^k \sum_{\underline{p}, |p| \leq k} \! \! |D^{\underline{p}}_{\bo x} f(\bo x + s \bo \xi)|
\sum_{\underline{p}', |p'| \leq k} \! \! |D^{\underline{p}'}_{\bo x} f(\bo x + s' \bo \xi')|\nonumber\\
&&\leq  c_k \, \langle \bo \xi \rangle^k\langle \bo \xi' \rangle^k \!\sum_{\underline{l}, |l|\leq k} \!\!  | D^{\underline{l}}_{\bo \xi} \widetilde{V}(\bo \xi)|
    \sum_{\underline{l}', |l'|\leq k} \!\!  | D^{\underline{l}'}_{\bo \xi'} \widetilde{V}(\bo \xi')|\langle \bo x + s \bo \xi \rangle^k\langle \bo x + s '\bo \xi'\rangle^k  \nonumber\\
&&\langle s \rangle^k\langle s' \rangle^k \sum_{\underline{p}, |p| \leq k} \! \! |D^{\underline{p}}_{\bo x} f(\bo x + s \bo \xi)|
\sum_{\underline{p}', |p'| \leq k} \! \! |D^{\underline{p}'}_{\bo x} f(\bo x + s' \bo \xi')|
\end{eqnarray}
where, in the last step,  we have used lemma \ref{lemmazeta}. Therefore
\ba\label{G3}
&&\left| \int \!\!  d\bo\xi \int \!\!  d\bo\xi'\,
G_{\varepsilon}(\bo\xi,\bo \xi',s,s';\bo x)
e^{\frac{i}{\ve}\Phi_j (\bo\xi,s,\hat{\bo u})} 
 e^{-\frac{i}{\ve}\Phi_j (\bo\xi',s',\hat{\bo u}')}\right| \nonumber\\
 &&\leq c_k \f{\ve^{2k}}{\Delta^{2k}}\langle s \rangle^k\langle s' \rangle^k 
 \int \!\! d\bo \xi \,  \int \!\! d\bo \xi' \,
  \langle \bo \xi \rangle^k\langle \bo \xi' \rangle^k \!\sum_{\underline{l}, |l|\leq k} \!\!  | D^{\underline{l}}_{\bo \xi} \widetilde{V}(\bo \xi)|
    \sum_{\underline{l}', |l'|\leq k} \!\!  | D^{\underline{l}'}_{\bo \xi'} \widetilde{V}(\bo \xi')| \nonumber\\
&&\langle \bo x + s \bo \xi \rangle^k\langle \bo x + s '\bo \xi'\rangle^k \sum_{\underline{p}, |p| \leq k} \! \! |D^{\underline{p}}_{\bo x} f(\bo x + s \bo \xi)|
\sum_{\underline{p}', |p'| \leq k} \! \! |D^{\underline{p}'}_{\bo x} f(\bo x + s' \bo \xi')|
\ea
Using (\ref{G3}) in (\ref{nonstaz2}) we have
\ba\label{nonstaz3}
&&\| I_j^{\ve} (t) (\Psi_0^{\ve} - \Psi_{0,j}^{\ve} ) \|^2 \nonumber\\
&&\leq c_k\,  \ve^{2k-2} \, t^{2k+2} 
 \int \!\! d\bo \xi \!\!  \int \!\! d\bo \xi' \,
  \langle \bo \xi \rangle^k\langle \bo \xi' \rangle^k \!\sum_{\underline{l}, |l|\leq k} \!\!  | D^{\underline{l}}_{\bo \xi} \widetilde{V}(\bo \xi)|
    \sum_{\underline{l}', |l'|\leq k} \!\!  | D^{\underline{l}'}_{\bo \xi'} \widetilde{V}(\bo \xi')| \nonumber\\
&&\int\!\! d\bo x \, \langle \bo x + s \bo \xi \rangle^k\langle \bo x + s '\bo \xi'\rangle^k \sum_{\underline{p}, |p| \leq k} \! \! |D^{\underline{p}}_{\bo x} f(\bo x + s \bo \xi)|
\sum_{\underline{p}', |p'| \leq k} \! \! |D^{\underline{p}'}_{\bo x} f(\bo x + s' \bo \xi')|
\ea
Finally we use  the Schwartz inequality in the last integral and we obtain
\be
\| I_j^{\ve} (t) (\Psi_0^{\ve} - \Psi_{0,j}^{\ve} ) \|^2 \leq c_k\,t^{2k+2}\|\widetilde{V}\|_{W_k^{k,1}}^2 \, \ve^{2k-2}
\ee
concluding the proof of the proposition \ref{prop3}.


\section{Estimate of the rest}
In this section we prove the proposition \ref{prop4}.
From (\ref{J}) we have
\begin{equation}\label{Resto1}
\|J^{\ve}(t)\Psi^{\ve}_0\|\leq    \sum_{\substack{k,l=1\\k\neq l}}^N\|J^{\ve}_{k,l}(t)\Psi^{\ve}_0\|+\sum_{k=1}^N \|J^{\ve}_k(t)\Psi^{\ve}_0\|
\end{equation}

\n
Let us estimate the first  term in the r.h.s of (\ref{Resto1}). Using (\ref{Jkl}), (\ref{Jklnm}), (\ref{Gnm}) and (\ref{thklnm}) we obtain
\begin{eqnarray}\label{J1}
&&\|J^{\ve}_{k,l}(t)\Psi^{\ve}_0\|^2
\nonumber\\
&&=\ve^3\sum_{\un,\um}\int d\bo x\left|\mathcal{J}^{\ve,k,l}_{\un,\um} (t,  \bo x) \right|^2\nonumber\\
&&=\f{ \mathcal N_{\ve}^2}{\ve^{2}}\sum_{\un,\um}\int d\bo x\left|
  \int_0^t \!\! ds \! \!\int_0^s \!\! d\sigma\;\!\! \int\!\!d\bo \xi \!\!
  \int\!\!d\bo \eta \int_{S^2} \!\!\!\!\! d\hat{\bo u} \; G_{\un,\um} (\bo \xi,\bo \eta, s,\sigma;\bo x)
   \; e^{ \f{i}{\ve}\Theta_{\un,\um}^{k,l} (\bo \xi, \bo \eta,  s,\sigma ,\hat{\bo u};\bo x)}\right|^2\nonumber\\
&&=\f{ \mathcal N_{\ve}^2}{\ve^{2}} \sum_{\un,\um}\int d\bo x\int_0^t \!\! ds \!
 \!\int_0^s \!\! d\sigma\;\!\! \int\!\!d\bo \xi \!\!\int\!\!d\bo \eta \int_{S^2} \!\!\!\!\!
 d\hat{\bo u} \;\int_0^t \!\! ds' \! \!\int_0^{s'} \!\! d\sigma'\;\!\! \int\!\!d\bo \xi' \!\!
 \int\!\!d\bo \eta'\int_{S^2} \!\!\!\!\! d\hat{\bo u}' \;\nonumber\\
&& e^{i (\bo \eta+\bo \xi) \cdot \bo x + i(\f{s}{2} \bo \eta^2+\f{\sigma}{2}\bo \xi^2)+
is\bo\eta\cdot\bo\xi} f(\bo x+\sigma\,\bo \xi +s\,\bo\eta)
e^{-\f{i}{\ve}\left[\bo \xi \cdot \bo a_l + \bo \eta \cdot \bo a_k - v_0\, \hat{\bo u}
\cdot (\bo x + \sigma\,\bo \xi + s\,\bo\eta)\right]}\nonumber\\
&& e^{-i (\bo \eta'+\bo \xi') \cdot \bo x - i(\f{s'}{2} \bo \eta^{'2}+\f{\sigma'}{2}\bo \xi^{'2})
-is'\bo\eta'\cdot\bo\xi'}\, f(\bo x+\sigma'\,\bo \xi' +s'\,\bo\eta')
e^{\f{i}{\ve}\left[\bo \xi' \cdot \bo a_l + \bo \eta' \cdot \bo a_k - v_0\, \hat{\bo u}'
\cdot (\bo x + \sigma'\,\bo \xi' +s'\,\bo\eta')\right]}\nonumber\\
&& \sum_{\un}e^{i|n| \f{s-s'}{\ve}}g_{\un,\uo}(\bo \eta)\bar{g}_{\un,\uo}(\bo \eta')\,
 \sum_{\um}e^{i|m| \f{\sigma-\sigma'}{\ve}}g_{\um,\uo}(\bo \xi)\bar{g}_{\um,\uo}(\bo \xi')
\end{eqnarray}
Taking into account formula (\ref{sumn1t0}) we compute the sum over $\un,\um$
\ba\label{J1}
&&\|J^{\ve}_{k,l}(t)\Psi^{\ve}_0\|^2
\nonumber\\
&&=  \f{ \mathcal N_{\ve}^2}{(2\pi)^6\ve^{2}}\int d\bo{x}\int_0^t \!\! ds \! \!
\int_0^t \!\! ds' \! \!\int_0^s \!\! d\sigma\!\! \int_0^{s'} \!\! d\sigma'\,e^{i\f{3}{2\ve}(s'-s)}
e^{i\f{3}{2\ve}(\sigma'-\sigma)} \int_{S^2} \!\!\!\!\! d\hat{\bo u} \; \int_{S^2} \!\!\!\!\! d\hat{\bo u}'
e^{-\f{i}{\ve}( \hat{\bo u}- \hat{\bo u}' )\cdot \bo x} \nonumber\\
&&\int\!\!d\bo \xi \!\! \int\!\!d\bo \xi' \!\!\int\!\!d\bo \eta \int\!\!d\bo \eta'\,A^{\ve}(\bo\xi,\bo\xi',
\bo\eta,\bo\eta', s, s', \sigma ,\sigma';\bo x)e^{\f{i}{\ve}\Theta_{kl}(\bo\xi,
\bo\eta, s, \sigma, \hat{\bo u} )}e^{-\f{i}{\ve}\Theta_{kl}(\bo\xi',
\bo\eta', s', \sigma', \hat{\bo u} )}
\ea
where
\ba
 &&A^{\ve}(\bo\xi,\bo\xi',\bo\eta,\bo\eta', s, s', \sigma ,\sigma'; \bo x)=
e^{i (\bo \eta+\bo \xi) \cdot \bo x + i(\f{s}{2} \bo \eta^2+\f{\sigma}{2}\bo \xi^2)+is\bo\eta\cdot\bo\xi}
 e^{-i (\bo \eta'+\bo \xi') \cdot \bo x - i(\f{s'}{2} \bo \eta^{'2}+\f{\sigma'}{2}\bo \xi^{'2})
 -is'\bo\eta'\cdot\bo\xi'}\nonumber\\\label{A}
 &&\widetilde{V}(\bo\eta)\overline{\widetilde{V}}(\bo\eta')
  \widetilde{V}(\bo\xi)\overline{\widetilde{V}}(\bo\xi')\zeta_{(s'-s)/\ve}(\bo\eta,\bo\eta')
\zeta_{(\sigma'-\sigma)/\ve}(\bo\xi,\bo\xi')
f(\bo x+\sigma\,\bo \xi +s\,\bo\eta)f(\bo x+\sigma'\,\bo \xi' +s'\,\bo\eta')\nonumber\\
\\
&&\Theta_{kl}(\bo\xi,
\bo\eta, s, \sigma, \hat{\bo u} )=v_0\bo \xi\cdot ( \sigma\,\hat{\bo u} -\tau_l \hat{\bo a}_l)+
v_0\bo \eta\cdot ( s\,\hat{\bo u} -\tau_k \hat{\bo a}_k)\ea
Let us consider the oscillatory integrals in the variables $\bo \xi, \bo\xi', \bo\eta, \bo \eta'$ in formula (\ref{J1}).
We observe  that the $|\nabla_{\bo\xi,\bo\eta} \Theta_{kl} |$ doesn't vanish for $\hat{\bo u}\in S^2$. In fact one has
\ba\label{stigrad1}
&&|\nabla_{\bo \xi,\bo \eta} \Theta_{kl} |^2 = \left(   v_0^2 |\sigma \hat{\bo u} - \tau_l 
\hat{\bo a}_l |^2 + v_0^2 |s \hat{\bo u} - \tau_k \hat{\bo a}_k |^2       \right) \geq 
\min_{\hat{ \bo u} \in S^2} \left(   v_0^2 |\sigma \hat{\bo u} - \tau_l 
\hat{\bo a}_l |^2 + v_0^2 |s \hat{\bo u} - \tau_k \hat{\bo a}_k |^2       \right) \nonumber\\
&&= \min \left\{   \min_{\hat{ \bo u} \in S^2 \setminus \mathcal C_k} \left(   v_0^2 |\sigma \hat{\bo u} - \tau_l 
\hat{\bo a}_l |^2 + v_0^2 |s \hat{\bo u} - \tau_k \hat{\bo a}_k |^2       \right)\! , \,  
 \min_{\hat{ \bo u} \in \mathcal C_k} \left(   v_0^2 |\sigma \hat{\bo u} - \tau_l 
\hat{\bo a}_l |^2 + v_0^2 |s \hat{\bo u} - \tau_k \hat{\bo a}_k |^2       \right) \right\} \nonumber\\
&&\geq \min \left\{   \min_{\hat{ \bo u} \in S^2 \setminus \mathcal C_k} \left(    v_0^2 |s \hat{\bo u} - \tau_k \hat{\bo a}_k |^2       \right)\! , \,  
 \min_{\hat{ \bo u} \in \mathcal C_k} \left(   v_0^2 |\sigma \hat{\bo u} - \tau_l 
\hat{\bo a}_l |^2        \right) \right\} \nonumber\\
&& \geq  \min \Big\{    v_0^2 \tau_k^2 \sin^2 \theta_0 , \,    v_0^2 \tau_l^2 \sin^2 \theta_0    \Big\} \geq  v_0^2 \tau_1^2 \sin^2 \theta_0 = \Delta^2
\ea
where in the last step we have used (\ref{cac}).  Hence
we can estimate the  integrals in the variables $\bo\xi,\bo \xi', \bo\eta,\bo\eta'$
exploiting a non-stationary phase argument.
With a repeated application of the identity (\ref{per_parti}) and using (\ref{stigrad1})
we have
\ba\label{j2}
&&\left|\int\!\!d\bo \xi \!\! \int\!\!d\bo \xi' \!\!\int\!\!d\bo \eta
\int\!\!d\bo \eta'\,A^{\ve}(\bo\xi,\bo\xi',
\bo\eta,\bo\eta', s, s', \sigma ,\sigma'; \bo x)e^{\f{i}{\ve}\Theta_{kl}(\bo\xi,
\bo\eta, s, \sigma, \hat{\bo u} )}e^{-\f{i}{\ve}\Theta_{kl}(\bo\xi',
\bo\eta', s', \sigma', \hat{\bo u} )}\right|\nonumber\\
&&\leq\frac{\ve^{2d}}{\Delta^{2d}}\!\!\!
\sum_{\substack{\um_1,\,\um_2,\,\um_1',\,\um_2'\\ |m_1|+|m_2|=d\\|m_1'|+|m_2'|=d}}\!
\int\!\!d\bo \xi \!\! \int\!\!d\bo \xi' \!\!\int\!\!d\bo \eta \int\!\!d\bo \eta'\,
\left|D^{\um_1}_{\bo\xi}D^{\um_2}_{\bo\eta}D^{\um'_1}_{\bo\xi'}D^{\um'_2}_{\bo\eta'}A^{\ve}(\bo\xi,\bo\xi',
\bo\eta,\bo\eta', s, s', \sigma ,\sigma'; \bo x)\right|\nonumber\\\ea
for any integer $d >0$. 
The next step is  to estimate the derivatives of $A^{\ve}$. Proceeding as in the previous section  we obtain
\ba\label{j4}
&&\left|\int\!\!d\bo \xi \!\! \int\!\!d\bo \xi' \!\!\int\!\!d\bo \eta
\int\!\!d\bo \eta'\,A^{\ve}(\bo\xi,\bo\xi',
\bo\eta,\bo\eta', s, s', \sigma ,\sigma'; \bo x)e^{\f{i}{\ve}\Theta_{kl}(\bo\xi,
\bo\eta, s, \sigma, \hat{\bo u} )}e^{-\f{i}{\ve}\Theta_{kl}(\bo\xi',
\bo\eta', s', \sigma', \hat{\bo u} )}\right|\nonumber\\
&&\leq c_d  \, \f{\ve^{2d}}{\Delta^{2d}}\langle s \rangle^d\langle s' \rangle^d\langle \sigma\rangle^d
\langle \sigma' \rangle^d
 \int \!\! d\bo \xi \,  \int \!\! d\bo \xi' \,
  \langle \bo \xi \rangle^d\langle \bo \xi' \rangle^d \!\sum_{\underline{l}, |l|\leq d} \!\!
   | D^{\underline{l}}_{\bo \xi} \widetilde{V}(\bo \xi)|
    \sum_{\underline{l}', |l'|\leq d} \!\!  | D^{\underline{l}'}_{\bo \xi'}\overline{ \widetilde{V}}(\bo \xi')| \nonumber\\
&& \int \!\! d\bo \eta \,  \int \!\! d\bo \eta' \,
  \langle \bo \eta \rangle^d\langle \bo \eta' \rangle^d \!\sum_{\underline{p}, |p|\leq d} \!\!
   | D^{\underline{p}}_{\bo \eta} \widetilde{V}(\bo \eta)|
    \sum_{\underline{p}', |p'|\leq d} \!\!  | D^{\underline{p}'}_{\bo \eta'}\overline{ \widetilde{V}}(\bo \eta')| \nonumber\\
&&\langle \bo x + \sigma \bo \xi+ s \bo\eta \rangle^d\langle \bo x + \sigma' \bo \xi'+ s '\bo\eta\rangle^d
 \sum_{\underline{q}, |q| \leq d} \! \! |D^{\underline{q}}_{\bo x} f(\bo x + \sigma \bo \xi+ s \bo\eta )|
\sum_{\underline{q}', |q'| \leq d} \! \! |D^{\underline{q}'}_{\bo x} f(\bo x + \sigma' \bo \xi'+ s '\bo\eta)|
\nonumber\\\ea
Then, substituting  (\ref{j4}) in (\ref{J1}) and  using the Schwartz inequality in the integral
 with respect to the variable $\bo x$, we have the following estimate of the first term in the r.h.s. of (\ref{Resto1})
\ba\|J^{\ve}_{k,l}(t)\Psi^{\ve}_0\|^2  \leq c_d\, t^{4d+4}\, \|\widetilde{V}\|_{W_d^{d,1}}^4\,  \ve^{2d-2}
\ea

\n
Let us analyze the second term in the r.h.s of (\ref{Resto1}). We write 
\be\label{j5}
\|J^{\ve}_{k}(t)\Psi^{\ve}_0\|\leq\|J^{\ve}_{k}(t)(\Psi^{\ve}_0-\Psi^{\ve}_{0,k})\|+
\|J^{\ve}_{k}(t)\Psi^{\ve}_{0,k}\|
\ee
From definition (\ref{PJ}) and from (\ref{Jk}) we have
\ba\label{j6}
&&\|J^{\ve}_{k}(t)(\Psi^{\ve}_0-\Psi^{\ve}_{0,k})\|^2 = \nonumber\\
&& \f{ \mathcal N_{\ve}^2}{\ve^{2}}
 \sum_{\un} \left|\sum_{\um}\int_0^t \!\! ds \! \!\int_0^s \!\! d\sigma\!\! \int\!\!d\bo \xi \!\!\int\!\!d\bo \eta \int_{S^2\backslash \mathcal{C}_k}
   \!\!\!\!\! d\hat{\bo u} \; L_{\un,\um} (\bo \xi,  \bo \eta, s,\sigma;\bo x)
   \; e^{ \f{i}{\ve} \Theta_{\un,\um}^{k} (\bo \xi, \bo \eta,  s,\sigma \hat{\bo u};\bo x)}\right|^2
\ea
and $L_{\un,\um},\Theta_{\un,\um}^{k}$ are defined by (\ref{Lnm}), (\ref{thknm}).  
The sum over $\um$ in (\ref{j6}) can be computed as follows
\ba\label{summn1}
&&\sum_{\um}e^{\f{i}{\ve}|m|( \sigma-s)}g_{\un,\um}(\bo \eta)
g_{\um,\uo}(\bo \xi)\nonumber\\
&&=
\frac{\widetilde{V}(\bo \xi)\overline{\widetilde{V}}(\bo \eta)}{(2\pi)^3}
\sum_{\um}\!\!\int d\bo y \phi_{\un}(\bo y)\phi_{\um}(\bo y)e^{-i\bo y\cdot \bo \eta}
\!\!\int d\bo y' (\bo y')\phi_{\uo}(\bo y')e^{i\bo y'\cdot \bo \xi}
e^{\f{i}{\ve}|m|( \sigma-s)}\nonumber\\
&&=
\frac{\widetilde{V}(\bo \xi)\overline{\widetilde{V}}(\bo \eta)}{(2\pi)^3}
\int d\bo y\,\overline{\phi_{\un}(\bo y)e^{-i\bo y\cdot \bo \eta}}\int d\bo y'
\sum_{\um}\phi_{\um}(\bo y)\phi_{\um}(\bo y')e^{i
\frac{|m|}{\ve} (\sigma-s) }\phi_{\uo}(\bo y')e^{-i\bo y'\cdot\bo \xi}\nonumber\\
&&
=\frac{\widetilde{V}(\bo\xi)\overline{\widetilde{V}}(\bo \eta)}{(2\pi)^3}
e^{i\frac{3}{2\varepsilon} (\sigma-s) }
\langle\phi_{\un},  e^{i \bo \eta \cdot (\cdot) } U\!\!\left(\frac{s-\sigma}{\ve}\right)
 e^{- i \bo \xi \cdot (\cdot) } \phi_{\uo} \rangle
\ea
therefore
\ba\label{summn2}
&&\sum_{\un}e^{\f{i}{\ve}|n|(s-s')}\sum_{\um}e^{\f{i}{\ve}|m|( \sigma-s)}g_{\un,\um}(\bo \eta)
g_{\um,\uo}(\bo \xi)\sum_{\um'}e^{-\f{i}{\ve}|m'|( \sigma'-s')}\bar{g}_{\un,\um'}(\bo \eta')
\bar{g}_{\um',\uo}(\bo \xi')\nonumber\\
 &&=\frac{\widetilde{V}(\bo\xi)\overline{\widetilde{V}}(\bo \eta)\widetilde{V}(\bo\xi')\overline{\widetilde{V}}(\bo \eta')e^{i\frac{3}{2\varepsilon} (\sigma-\sigma') }}{(2\pi)^6}\nonumber\\
&& \cdot
\langle  e^{i \bo \eta' \cdot (\cdot) } U\!\!\left(\frac{s'-\sigma'}{\ve}\right)
 e^{- i \bo \xi' \cdot (\cdot) } \phi_{\uo},U\!\!\left(\frac{s'-s}{\ve}\right)
  e^{i \bo \eta\cdot (\cdot) } U\!\!\left(\frac{s-\sigma}{\ve}\right)
 e^{- i \bo \xi \cdot (\cdot) } \phi_{\uo} \rangle\nonumber\\
 &&=\frac{\widetilde{V}(\bo\xi)\overline{\widetilde{V}}(\bo \eta)\widetilde{V}(\bo\xi')\overline{\widetilde{V}}(\bo \eta')e^{i\frac{3}{2\varepsilon} (\sigma-\sigma') }}{(2\pi)^6}
\zeta_{s_1, s_2, s_3}(\bo \xi',\bo \eta',\bo \eta,\bo \xi) \ea
where $s_1= (\sigma' - s')/ \ve, \; s_2= (s'-s)/ \ve,\; s_3=(s-\sigma)/ \ve$. 
Using (\ref{summn2}) we have
\ba
&&\|J^{\ve}_{k}(t)(\Psi^{\ve}_0-\Psi^{\ve}_{0,k})\|^2=\nonumber\\
&&\f{ \mathcal N_{\ve}^2}{(2\pi)^6\ve^{2}}\!\!\int d\bo x \!\!\int_0^t \!\! ds \! \!
\int_0^s \!\! d\sigma\!\!\int_0^t \!\! ds' \! \!
\int_0^s \!\! d\sigma'\,e^{\f{3i}{2\ve}(\sigma-\sigma')}\!\!\int_{S^2\backslash \mathcal{C}_k} \!\!\!\!\!
d\hat{\bo u} \int_{S^2\backslash \mathcal{C}_k} \!\!\!\!\! d\hat{\bo u'}\;e^{\f{i}{\ve}v_0(\hat{\bo u}-\hat{\bo u}')\cdot \bo x}
\nonumber\\&& \cdot\int\!\!d\bo \xi \!\!\int\!\!d\bo \eta \!\!\int\!\!d\bo \xi' \!\!\int\!\!d\bo \eta'\; L^{\ve} (\bo \xi,\bo \xi',  \bo \eta,\bo \eta', s, s',\sigma,\sigma';\bo x) \;e^{\f{i}{\ve}
\Theta_{k} (\bo \xi,\bo \eta,  s,\sigma, \hat{\bo u})}e^{-\f{i}{\ve}
\Theta_{k} (\bo \xi',\bo \eta',  s',\sigma', \hat{\bo u}')}
\ea
where \ba
&&L^{\ve}(\bo \xi,\bo \xi',  \bo \eta,\bo \eta', s, s',\sigma,\sigma';\bo x) 
 = e^{i (\bo \eta+\bo \xi) \cdot \bo x +
i(\f{s}{2} \bo \eta^2+\f{\sigma}{2}\bo \xi^2)+is\bo\eta\cdot\bo\xi}
 f(\bo x+\sigma\,\bo \xi +s\,\bo\eta)
\nonumber\\
&& \cdot e^{-i (\bo \eta'+\bo \xi') \cdot \bo x -
i(\f{s}{2} \bo \eta^{'2}+\f{\sigma'}{2}\bo \xi^{'2})-is'\bo\eta'\cdot\bo\xi'}f(\bo x+\sigma'\,\bo \xi' +s'\,\bo\eta')\widetilde{V}(\bo\xi)
\overline{\widetilde{V}}(\bo \xi')\widetilde{V}(\bo\eta)
\overline{\widetilde{V}}(\bo \eta')\zeta_{s_1, s_2, s_3}(\bo \xi',\bo \eta',\bo \eta,\bo \xi)   \nonumber\\
&&\Theta_{k} (\bo \xi,\bo \eta,  s,\sigma, \hat{\bo u};\bo x)\!\!=
  v_0\, (\sigma \hat{\bo u}-\tau_k\hat{\bo a_k}) \cdot\,\bo \xi
  +   v_0\, (s \hat{\bo u}-\tau_k\hat{\bo a_k}) \cdot\,\bo \eta
\ea
We observe that the $|\nabla_{\bo\xi,\bo\eta}\Theta_{k}|$   doesn't vanish in $S^2\backslash \mathcal{C}_k$. In fact we have
\be|\nabla_{\bo\xi,\bo\eta}\Theta_{k}|\geq v_0\tau_k\sin \theta_0\geq  v_0\tau_1\sin \theta_0\equiv \Delta
\ee 
Hence we can estimate the  integrals in the variables $\bo\xi,\bo \xi', \bo\eta, \bo\eta'$
exploiting a non-stationary phase argument along the same line of the previous case
(see (\ref{j2}),(\ref{j4})) and we
find 
\be\|J^{\ve}_{k}(t)(\Psi^{\ve}_0-\Psi^{\ve}_{0,k})\|^2\leq c_d \;t^{4d+4}\,\|\tilde{V}\|_{W_d^{d,1}}^4\, \ve^{2d-2}\ee
for any integer $d>0$.

\n
Finally we consider the last term in (\ref{j5}).  We have 
\ba\label{j7}
&&\|J^{\ve}_{k}(t)\Psi^{\ve}_{0,k}\|^2 =  \f{ \mathcal N_{\ve}^2}{\ve^{2}}\sum_{\un}
  \left|\sum_{\um} \!\!\int_0^t \!\! ds \! \!\int_0^s \!\! d\sigma\!\! \int\!\!d\bo \xi \!\!\int\!\!d\bo \eta \int_{ \mathcal{C}_k}
   \!\!\!\!\! d\hat{\bo u} \; L_{\un,\um} (\bo \xi,  \bo \eta, s,\sigma;\bo x)
   \; e^{ \f{i}{\ve} \Theta_{\un,\um}^{k} (\bo \xi, \bo \eta,  s,\sigma \hat{\bo u};\bo x)}\right|^2\nonumber\\
   &&
\ea
We rewrite  the integral
in the variables $s, \sigma, \bo\xi, \bo\eta$ in a more convenient form.
In particular by the change of variables $\bo\xi\rightarrow\mathcal{R}_k\bo\xi$, $\bo\eta\rightarrow\mathcal{R}_k\bo\eta$,
$\hat{\bo u}\rightarrow \mathcal{R}_k\hat{\bo u} $ with  $\mathcal{R}_k$  defined in (\ref{mr}), we
obtain
\ba \label{j8}&&\int_0^t \!\! ds \! \!\int_0^s \!\! d\sigma\!\! \int\!\!d\bo \xi \!\!\int\!\!d\bo \eta \int_{ \mathcal{C}_k}
   \!\!\!\!\! d\hat{\bo u} \; L_{\un,\um} (\bo \xi,  \bo \eta, s,\sigma;\bo x)
   \; e^{ \f{i}{\ve} \Theta_{\un,\um}^{k} (\bo \xi, \bo \eta,  s,\sigma \hat{\bo u};\bo x)}\nonumber\\
   &&=\int_0^t \!\! ds \! \!\int_0^s \!\! d\sigma\!\! \int\!\!d\bo \xi \!\!\int\!\!d\bo \eta \int_{\mathcal{ C}_0}
   \!\!\!\!\! d\hat{\bo u} \; L_{\un,\um} (\mathcal{R}_k^{-1}\bo \xi,  \mathcal{R}_k^{-1}\bo \eta, s,\sigma;\bo x)
   \; e^{ \f{i}{\ve} \Theta_{\un,\um}^{k} (\mathcal{R}_k^{-1}\bo \xi, \mathcal{R}_k^{-1}\bo \eta,  s,\sigma , \hat{\bo u};\bo x)}\nonumber\\
   \ea
   where
   \ba
&&L_{\un,\um} (\mathcal{R}_k^{-1}\bo \xi,  \mathcal{R}_k^{-1}\bo \eta, s,\sigma;\bo x)\nonumber\\
&&=
 e^{i (\bo \eta+\bo \xi) \cdot \bo x^k + i(\f{s}{2} \bo \eta^2+\f{\sigma}{2}\bo \xi^2)+is\bo\eta\cdot\bo\xi}g_{\un,\um}(\mathcal{R}_k^{-1}
\bo \eta)g_{\um,\uo}(\mathcal{R}_k^{-1}\bo \xi) f(\bo x^k+\sigma\,\bo \xi +s\,\bo\eta) \\
&&\Theta_{\un,\um}^{k} (\mathcal{R}_k^{-1}\bo \xi, \mathcal{R}_k^{-1}\bo \eta,  s,\sigma, \hat{\bo u};\bo x) \nonumber\\
&& =
 - (\xi_3+\eta_3) |\bo a_k| + v_0\, \hat{\bo u} \cdot (\bo x^k +\sigma\,\bo \xi +
  s\,\bo\eta) +(|n| -|m|)s+|m|\sigma
 \ea
and $\bo x^k\equiv \mathcal{R}_k \bo x$.  Moreover we parametrize the unit vector $\hat{\bs{u}} \in \mathcal C_0$ as follows
\be
\hat{\bs{u}} = \left( \mu, \nu, \sqrt{1-\mu^2  -\nu^2} \right), \;\;\;\; (\mu, \nu) \in \mathcal{D}_0 \equiv \left\{ (a,b) \in \erre^2, \; a^2 + b^2 < \sin^2 \theta_0 \right\}
\ee
Therefore the integral (\ref{j8}) is rewritten as
\ba\label{j9}
&&\int_0^t \!\! ds \! \!\int_0^s \!\! d\sigma\!\! \int\!\!d\bo \xi \!\!\int\!\!d\bo \eta \int_{ \mathcal{C}_0}
   \!\!\!\!\! d\hat{\bo u} \; L_{\un,\um} (\bo \xi,  \bo \eta, s,\sigma;\bo x)
   \; e^{ \f{i}{\ve} \Theta_{\un,\um}^{k} (\bo \xi, \bo \eta,  s,\sigma, \hat{\bo u};\bo x)}=\nonumber\\
   &&\int_0^t \!\! ds \! \!\int_0^s \!\! d\sigma\!\! \int\!\!d\bo \xi \!\!\int\!\!d\bo \eta \int_{\mathcal{ D}_0}
   \!\!\!\!\! d\mu d\nu \; \widetilde{L}_{\un,\um} (\bo \xi,  \bo \eta, s,\sigma, \mu, \nu ;\bo x)
   \; e^{ \f{i}{\ve} \widetilde{\Theta}_{\un,\um}^{k} (\bo \xi, \bo \eta,  s,\sigma , \mu , \nu;\bo x)}\nonumber\\
\ea
where
\ba\label{theta}&&
\widetilde{L}_{\un,\um} (\bo \xi, \bo \eta, s,\sigma,\mu, \nu;\bo x)=\frac{1}{\sqrt{1-\mu^2-\nu^2}}L_{\un,\um} (\mathcal{R}_k^{-1}\bo \xi,  \mathcal{R}_k^{-1}\bo \eta, s,\sigma;\bo x) \\
&& \widetilde{\Theta}_{\un,\um}^{k} (\bo \xi, \bo \eta,  s,\sigma , \mu , \nu;\bo x)= - (\xi_3+\eta_3) |\bo a_k| + v_0\, \mu( x^k_1 +\sigma\, \xi_1 +
  s\,\eta_1) \nonumber\\
  && +\, v_0\, \nu ( x^k_2 \!+\! \sigma\, \xi_2 \!+\! 
  s\,\eta_2 )+ v_0\sqrt{1 \!-\! \mu^2 \!-\! \nu^2}( x^k_3 \!+\! \sigma\, \xi_3 \!+\!
  s\,\eta_3 ) +(|n| -|m|)s+|m|\sigma
\ea
Let us introduce the following linear change of coordinates
\ba\label{chco1}
&&(\mu,\nu,s) = L_{\ve}^1 (z_1,z_2,z_3) \equiv L_{\ve}^1 \bs{z} \;\;\;\;\, \sigma=L_{\ve}^2 p\\
&&\mu= \f{\ve}{v_0 \tau_k} z_1, \;\;\; \nu= \f{\ve}{v_0 \tau_k} z_2, \;\;\; s= \tau_k +\f{\ve}{v_0} z_3\;\;\;
\sigma=\tau_k +\f{\ve}{v_0} p\ea
The domain of integration in the variables $\bs{z}, p$ is
\ba
&&\Lambda_{\ve} = \left\{ \bs{z} \in \erre^3\, p\in \erre\,|\, z_1^2+z_2^2 < \ve^{-2}v_0^2\tau_k^2 \sin^2 \theta_0,\; -\ve^{-1} v_0 \tau_k < z_3, p < \ve^{-1}v_0 (t-\tau_k)  \right\}\;\;\;\;\;\;
\ea
Using (\ref{chco1}), (\ref{j9}) in (\ref{j7}) we obtain
\ba\label{j11}
&&\|J^{\ve}_{k}(t)\Psi^{\ve}_{0,k}\|^2 =  \f{ \mathcal{N}_{\ve}^2}{\ve^{2}}\sum_{\un}
  \left|\sum_{\um}\int_0^t \!\! ds \! \!\int_0^s \!\! d\sigma\!\! \int\!\!d\bo \xi \!\!\int\!\!d\bo \eta \int_{ \mathcal{C}_k}
   \!\!\!\!\! d\hat{\bo u} \; L_{\un,\um} (\bo \xi,  \bo \eta, s,\sigma;\bo x)
   \; e^{ \f{i}{\ve} \Theta_{\un,\um}^{k} (\bo \xi, \bo \eta,  s,\sigma \hat{\bo u};\bo x)}\right|^2\nonumber\\
&&=\frac{\ve^6\mathcal{N}_{\ve}^2}{v_0^8\tau_k^4} \sum_{\un}\sum_{\um, \um'}
\int_{\Lambda_{\ve}} d \bo z\,dp\!\! \int\!\!d\bo \xi \!\!\int\!\!d\bo \eta \!
 \int_{\Lambda_{\ve}} d \bo z'\, dp' \!\! \int\!\!d\bo \xi' \!\!\int\!\!d\bo \eta'
 \widetilde{L}_{\un,\um} (\bo \xi,  \bo \eta, L_{\ve}^1\bo z,L_{\ve}^2 p;\bo x)\nonumber\\
&&\cdot  \, \overline{\widetilde{L}}_{\un,\um'} (\bo \xi',  \bo \eta', L_{\ve}^1\bo z', L_{\ve}^2p;\bo x)\,e^{i\bo z\cdot \bo \eta +ip\xi_3}
e^{iz_1\left(\f{x_1^k}{\tau_k}+\xi_1\right)+iz_2\left(\f{x_2^k}{\tau_k}+\xi_2\right)+
  i z_3\f{|n|-|m|}{v_0}+i p\f{|m|}{v_0}+iA^{\ve}(\bo z, p, \xi_3,\eta_3; \bo x) }\nonumber\\
&&\cdot  \, e^{-i\bo z'\cdot \bo\eta'-ip'\xi_3'}
 e^{ -   iz_1'\left(\f{x_1^k}{\tau_k}+\xi_1'\right)-iz_2'\left(\f{x_2^k}{\tau_k}+\xi_2'\right)-
  i z_3'\f{|n|-|m|'}{v_0}-i p'\f{|m|'}{v_0}-i B^{\ve}(\bo z',p',\xi_3',\eta_3'; \bo x) }
  \ea
  where
  \ba B^{\ve}(\bo z,p,\xi_3,\eta_3; \bo x)=\frac{\sqrt{1 \!-\! \left( \! \f{\ve z_1}{v_0 \tau_k} \! \right)^2 \!-\! \left(\! \f{\ve z_2}{v_0 \tau_k}\! \right)^2}-1}{\ve^2}(v_0x_3^k+|\bo a_k|\xi_3+|\bo a_k|\eta_3+\ve z_3\eta_3+\ve p\xi_3)
  \ea
Now we  compute the sum over $\un, \um , \um'$ exploiting (\ref{summn1}), (\ref{summn2}) and we obtain
\ba\label{j11}
&&\|J^{\ve}_{k}(t)\Psi^{\ve}_{0,k}\|^2 \equiv\frac{\ve^6\mathcal{N}_{\ve}^2}{(2\pi)^6 v_0^8\tau_k^4}
\int_{\Lambda_{\ve}}\!\! d \bo zdp\int_{\Lambda_{\ve}} \!\! d \bo z'dp\, E_k (\bo z,p, \bo z',p')\ea
where
\ba &&E_k (\bo z,p, \bo z',p')=
\int\!\!d\bo \xi \!\!\int\!\!d\bo \eta \!
  \int\!\!d\bo \xi' \!\!\int\!\!d\bo \eta'
\zeta_{b_1,b_2,b_3}(\mathcal{R}_k^{-1}\bo\xi',
\mathcal{R}_k^{-1}\bo\eta',\mathcal{R}_k^{-1}\bo\eta,\mathcal{R}_k^{-1}\bo\xi)\nonumber\\
&&e^{\frac{3i}{2v_0}(p-p')}\widetilde{V}(\mathcal{R}_k^{-1}\bo\xi)\overline{\widetilde{V}}
(\mathcal{R}_k^{-1}\bo\xi')
\widetilde{V}(\mathcal{R}_k^{-1}\bo\eta)\overline{\widetilde{V}}(\mathcal{R}_k^{-1}\bo\eta')
 e^{i (\bo \eta+\bo \xi) \cdot \bo x^k} e^{-i (\bo \eta'+\bo \xi') \cdot \bo x^k}\nonumber\\
 &&e^{ i\f{\tau_k}{2} (\bo \eta^2+\bo \xi^2)+i\tau_k\bo\eta\cdot\bo\xi}
 e^{ -i\f{\tau_k}{2} (\bo \eta^{'2}+\bo \xi^{'2})-i\tau_k\bo\eta'\cdot\bo\xi'}
 e^{ i\f{\ve}{2v_0} (z_3\bo \eta^2+p\bo \xi^2)+i\frac{\ve z_3 p}{v_0}\bo\eta\cdot\bo\xi}
 e^{ i\f{\ve}{2v_0} (z_3'\bo \eta^{'2}+p'\bo \xi^{'2})+i\frac{\ve z_3' p'}{v_0}\bo\eta'\cdot\bo\xi'}
 \nonumber\\
 &&f\left(\bo x^k+\tau_k\bo\eta+\tau_k\bo \xi +\frac{\ve}{v_0}z_3\bo\eta+\frac{\ve}{v_0}p\bo\xi\right) f\left(\bo x^k+\tau_k\bo\eta'+\tau_k\bo \xi' +\frac{\ve}{v_0}z_3'\bo\eta'+\frac{\ve}{v_0}p'\bo\xi'\right)
\nonumber\\
&&  e^{i\bo z\cdot \bo\eta+ip\xi_3}
e^{iz_1\left(\f{x_1^k}{\tau_k}+\xi_1\right)+iz_2\left(\f{x_2^k}{\tau_k}+\xi_2\right)
+i B^{\ve}(\bo z, p, \xi_3,\eta_3; \bo x) }\nonumber\\
&&e^{-i\bo z'\cdot \bo \eta-ip'\xi_3'}
 e^{ -   iz_1'\left(\f{x_1^k}{\tau_k}+\xi_1'\right)-iz_2'\left(\f{x_2^k}{\tau_k}+\xi_2'\right)-
i B^{\ve}(\bo z', p', \xi_3',\eta_3'; \bo x) }\ea
and $b_1=(p'-z_3')/v_0,\, b_2=(z_3'-z_3)/{v_0}, \,b_3=(z_3-p)/{v_0}$.\newline
It remains to show that the integrals  in (\ref{j11}) are bounded. 
Exploiting the  identity                                       
\ba  \partial^{4}_{\eta_l}\partial^4_{\xi_3}\left[e^{i\bo z\cdot \bo\eta+ip\xi_3}\right]  \partial^{4}_{\eta_l'}\partial^4_{\xi_3'}\left[e^{-i\bo z'\cdot \bo \eta-ip'\xi_3'}
\right]=
z_l^4p^4z_l^{'4}p^{'4}\left[e^{i\bo z\cdot \bo\eta+ip\xi_3}\right]
\left[e^{-i\bo z'\cdot \bo \eta-ip'\xi_3'}\right]\;\;\,l=1,2,3
\ea
we can integrate by parts and  we have 
\ba\label{pi}\langle \bo z\rangle^4\langle \bo z'\rangle^4 \langle p\rangle^4\langle p'\rangle^{4}| E_k (\bo z,p, \bo z',p')|\leq c\,\|\widetilde{V}\|_{W_4^{4,1}}^4\ea
Finally we use the estimate (\ref{pi}) in (\ref{j11}) and we find
\ba\label{j11}
&&\|J^{\ve}_{k}(t)\Psi^{\ve}_{0,k}\|^2 \leq c \,\ve^6
\int_{\Lambda_{\ve}}\!\! d \bo zdp\int_{\Lambda_{\ve}} \!\! d \bo z'dp\, E_k (\bo z,p, \bo z',p')\nonumber\\
&&\leq c \,\ve^6\|\widetilde{V}\|_{W_4^{4,1}}^4\int\!\! d \bo z d\bo z' \f{1}{\langle \bo z\rangle^4\langle \bo z'\rangle^4}\int\!\! dp dp'\f{1}{\langle p\rangle^4\langle p'\rangle^{4}}  \leq c\,\ve^6\|\widetilde{V}\|_{W_4^{4,1}}^4\ea
concluding the proof of the proposition.




\n

\n

\n



\newpage


\begin{thebibliography}{99}

\bibitem[AFFT]{afft} Adami R., Figari R., Finco D., Teta A., On the asymptotic dynamics of a quantum system composed by heavy and light particles.  {\em Comm.   Math. Phys.} {\bf 268}, no. 3, 819-852, 2006.



\bibitem[BGJKS]{bgjks} Blanchard Ph., Giulini D., Joos E., Kiefer C., 
Stamatescu I.-O.  eds., \emph{Decoherence: Theoretical, Experimental
and Conceptual Problems}, Lect. Notes  Phys. {\bf 538}, Springer, 2000.

\bibitem[BH]{bh} Bleinstein N., Handelsman R.A., \emph{Asymptotic Expansions of Integrals}, Dover Publ.,  1975.

\bibitem[CCF]{ccf} Cacciapuoti C.,Carlone R., Figari R., Decoherence induced by scattering:
a three dimensional model. {\em J. Phys. A: Math. Gen.}, {\bf 38}, n. 22, 4933-4946, 
2005.


\bibitem[CR]{cr} Comberscure M., Robert D., {\em Coherent States and Applications in Mathematical Physics}, Springer, 2012.

\bibitem[DFT]{dft} Dell'Antonio G., Figari R., Teta A., A time dipendent perturbative analysis
for a quantum particle in a cloud chamber.  {\em Ann.  H.  Poincare'} {\bf 11}, n. 3, 539-564, 2010.

\bibitem[F]{f} Fedoryuk M.V., The stationary phase method and pseudodifferential operators. {\em Usp. Mat. Nauk} {\bf 26}, n. 1, 67-112, 1971.

\bibitem[FigT]{figt} Figari R., Teta A., Emergence of classical trajectories in  quantum systems: 
the  cloud chamber problem in the analysis of Mott (1929). {\em Archive for History of Exact Sciences}, {\bf 67}, n. 2,  215-234, 2013.


\bibitem[FinT]{fint} Finco D.,Teta A., Asymptotic expansion for the wave function in a one-dimensional
model of inelastic interaction. {\em J. Math. Phys.}  {\bf 52}, 022103, 2011.



\bibitem[GJKKSZ]{gjkksz} Giulini D., Joos E., Kiefer C., Kupsch J., Stamatescu I.-O., Zeh
H.D., {\em Decoherence and the Appearance of a Classical World in
Quantum Theory}, Springer, 1996.

\bibitem[H]{h} Hornberger K., Introduction to Decoherence Theory. In:   Buchleitner A.,  Viviescas C.,  Tiersch M. eds., 
{\em Entanglement and Decoherence. Foundations and Modern Trends}.  
 Lect. Notes Phys. {\bf 768}, 221-276, Springer, 2009.

\bibitem[M]{m} Mott  N.F., The wave mechanics of $\alpha$-ray tracks. {\em Proc. R. Soc.
Lond. A}, {\bf  126}, 79-84, 1929.

\bibitem[R]{r} Robert D.,  Semi-classical approximation in quantum mechanics. A survey of old and recent mathematical results. {\em Helv. Phys. Acta}, {\bf 71}, 44-116 (1998).
\end{thebibliography}
\end{document}